\documentclass[11pt,a4]{article}
\pdfoutput=1

\usepackage[T1]{fontenc}
\usepackage{microtype}
\usepackage{fullpage}
\usepackage{authblk}
\usepackage{amsthm}
\usepackage{amssymb}
\usepackage[english]{babel}
\usepackage[utf8]{inputenc}
\usepackage{amsmath}
\usepackage{amsfonts}
\usepackage{graphicx}
\usepackage{mathtools}
\usepackage{enumerate}
\usepackage{thmtools}
\usepackage{url}

\newcommand\defeq{\stackrel{\mathclap{\normalfont{\mbox{\text{\tiny{def}}}}}}{=}}

\newtheorem{definition}{Definition}[section]
\newtheorem{lemma}[definition]{Lemma}
\newtheorem{theorem}[definition]{Theorem}

\newtheorem{property}[definition]{Property}

\newtheorem{observation}[definition]{Observation}
\newtheorem{corollary}[definition]{Corollary}

\graphicspath{{./figures/}}

\newcommand{\bigo}{\mathcal{O}}

\newcommand{\reach}[2]{\ensuremath{#1 {\leadsto} #2}}
\newcommand{\nreach}[2]{\ensuremath{#1 {\not\leadsto} #2}}

\newcommand{\tedp}[2]{\ensuremath{#1 {\leadsto}_{\text{2e}} #2}}
\newcommand{\ntedp}[2]{\ensuremath{#1 {\not\leadsto}_{\text{2e}} #2}}

\newcommand{\tvdp}[2]{\ensuremath{#1 {\leadsto}_{\text{2v}} #2}}
\newcommand{\ntvdp}[2]{\ensuremath{#1 {\not\leadsto}_{\text{2v}} #2}}

\newcommand{\closure}[1]{{#1}^{\tedp{}{}}}

\newcommand{\lrepr}{\mathrm{repr}_L}
\newcommand{\rrepr}{\mathrm{repr}_R}
\newcommand{\repr}{\mathrm{repr}}
\newcommand{\auxa}{H}
\newcommand{\auxb}{H'}

\newcommand{\bitwiseOR}{\vee}
\newcommand{\bitwiseAND}{\wedge}
\newcommand{\bitwiseBigOr}{\bigvee}

\usepackage[noend, noline, ruled, linesnumbered]{algorithm2e}
\SetAlFnt{\normalsize}
\DontPrintSemicolon
\SetKwComment{Comment}{$\triangleright$\ }{}
\SetKwProg{Def}{def}{:}{}

\SetKwFunction{KwClosure}{closure}
\SetKwFunction{KwRecovery}{recovery}
\SetKwFunction{KwEncode}{encode}
\SetKwFunction{KwDecode}{decode}
\SetKwFunction{KwMul}{mul}
\SetKwFunction{KwClosureSCC}{closureSCC}
\SetKwFunction{KwStrongConnect}{stronglyConnected}
\SetKwFunction{KwClosureDAG}{closureDAG}
\SetKwData{KwFalse}{false}
\SetKwData{KwTrue}{true}
\SetKwData{KwUnset}{unset}
\SetKwData{KwLeft}{left}
\SetKwData{KwRight}{right}
\SetKwData{KwUnset}{undefined}

\title{\bf All-Pairs $2$-Reachability in $\bigo(n^{\omega}\log n)$ Time.
}

\author[1]{Loukas~Georgiadis}
\author[2]{Daniel~Graf}
\author[3]{Giuseppe~F.~Italiano\thanks{Partially supported by MIUR, the Italian Ministry of Education, University and Research, under Project AMANDA (Algorithmics for MAssive and Networked DAta).}}
\author[3]{Nikos~Parotsidis}
\author[2]{Przemys\l{}aw~Uzna\'nski}
\affil[1]{
University of Ioannina, Greece.\\
\texttt{loukas@cs.uoi.gr}
}
\affil[2]{Department of Computer Science,
ETH Z\"urich, Switzerland.\\
	\texttt{daniel.graf@inf.ethz.ch}, \texttt{przemyslaw.uznanski@inf.ethz.ch}
}
\affil[3]{University of Rome Tor Vergata, Italy.\\
	\texttt{giuseppe.italiano@uniroma2.it}, \texttt{nikos.parotsidis@uniroma2.it}
}

\date{}
\begin{document}

\maketitle

\begin{abstract}
In the $2$-reachability problem we are given a directed graph $G$ and we wish to determine if there are two (edge or vertex) disjoint paths from $u$ to $v$, for a given pair of vertices $u$ and $v$. In this paper, we present an algorithm that computes $2$-reachability information for all pairs of vertices in $\bigo(n^{\omega}\log n)$ time, where $n$ is the number of vertices and $\omega$ is the matrix multiplication exponent.
Hence, we show that the running time of all-pairs $2$-reachability is only within a $\log$ factor of transitive closure.

Moreover, our algorithm produces a witness (i.e., a separating edge or a separating vertex) for all pair of vertices where $2$-reachability does not hold.
By processing these witnesses, we can compute all the edge- and vertex-dominator trees of $G$ in $\bigo(n^2)$ additional time,
which in turn enables us to answer various connectivity queries in $\bigo(1)$ time. For instance, we can test in constant time if there is a path from $u$ to $v$ avoiding an edge $e$, for any pair of query vertices $u$ and $v$, and any query edge $e$, or if there is a path from $u$ to $v$ avoiding a vertex $w$, for any query vertices $u$, $v$, and $w$.
\end{abstract}

\section{Introduction}
\label{sec:intro}

The \emph{all-pairs reachability problem} consists of preprocessing a directed graph (digraph) $G=(V,E)$ so that we can
answer queries that ask if a vertex $y$ is reachable from a vertex $x$. This problem has many applications, including databases, geographical information systems, social networks, and bioinformatics~\cite{Jin:2008}.
A classic solution to this problem is to compute the transitive closure matrix of $G$, either by performing a graph traversal (e.g., depth-first or breadth-first search) once per each vertex as source, or via matrix multiplication.
For a digraph with $n$ vertices and $m$ edges, the former solution runs in $\bigo(mn)$ time, while the latter runs in $\bigo(n^{\omega})$, where $\omega$
is the matrix multiplication exponent~\cite{matrix_mult:cw,LeGall:2014,Williams:2012}.
Here we study a natural generalization of the all-pairs reachability problem, that we refer to as \emph{all-pairs $2$-reachability}, where we wish to preprocess $G$ so that we can answer fast the following type of queries:
For a given vertex pair $x, y \in V$, are there two edge-disjoint (resp., internally vertex-disjoint) paths from $x$ to $y$?
Equivalently, by Menger's theorem~\cite{menger}, we ask if there is an edge $e \in E$ (resp., a vertex $z \in V$) such that there is no path from $x$ to $y$ in $G \setminus e$ (resp., $G \setminus z$). We call such an edge (resp., vertex) \emph{separating} for the pair $x$, $y$.

One solution to the all-pairs $2$-reachability problem is to compute all the dominator trees of $G$, with each vertex as source. The dominator tree of $G$ with start vertex $s$ is a tree
rooted at $s$, such that a vertex $v$ is an ancestor of a vertex $w$
if and only if all paths from $s$ to $w$ include $v$~\cite{domin:lt}.
All the separating edges and vertices for a pair $s$, $v$, appear on the path from $s$ to $v$ in the dominator tree rooted at $s$, in the same order as they appear in any path from $s$ to $v$ in $G$.
Given all the dominator trees, we can process them to compute the $2$-reachability information for all pairs of vertices (see Section \ref{sec:applications}). Since a dominator tree can be computed in $\bigo(m)$ time~\cite{domin:ahlt,dominators:bgkrtw}, the overall running time of this algorithm is $\bigo(mn)$.

\paragraph{Our Results.}
In this paper, we show how to beat the $\bigo(nm)$ bound for dense graphs. Specifically, we present an algorithm that computes $2$-reachability information for all pairs of vertices in
$\bigo(n^{\omega})$ time in a strongly connected digraph, and in  $\bigo(n^{\omega} \log{n})$ time in a general digraph.
Hence, we show that the running time of all-pairs $2$-reachability is only within a $\log$ factor of transitive closure.
This result is tight up to a $\log$ factor, since it can be shown that all-pairs $2$-reachability is at least as hard as computing the transitive closure,
which is asymptotically equivalent to Boolean matrix multiplication
\cite{fischer1971boolean}.
Moreover, our algorithm produces a witness (separating edge or separating vertex) whenever $2$-reachability does not hold.
By processing these witnesses, we can find all the dominator trees of $G$ in $\bigo(n^2)$ additional time.
Thus, we also show how
to compute all the dominator trees of a digraph in $\bigo(n^{\omega} \log{n})$ time (in $\bigo(n^{\omega})$ time if the graph is strongly connected),
which improves the previously known $\bigo(mn)$ bound for dense graphs.
This in turn enables us to answer various connectivity queries in $\bigo(1)$ time. For instance, we can test in $\bigo(1)$ time if there is a path from $u$ to $v$ avoiding an edge $e$, for any pair of query vertices $u$ and $v$, and any query edge $e$, or if there is a path from $u$ to $v$ avoiding a vertex $w$, for any query vertices $u$, $v$, and $w$.
We can also report all the edges or vertices that appear in all paths from $u$ to $v$, for any query vertices $u$ and $v$.

\paragraph{Related Work.}
To the best of our knowledge, ours is the first work that considers the all-pairs $2$-reachability problem and gives a fast algorithm for it.
In recent work Georgiadis et al.~\cite{GIP17:SODA}
investigate the effect of an edge or a vertex failure in
a digraph $G$ with respect to strong connectivity. Specifically, they show how to preprocess
$G$ in $\bigo(m+n)$ time in order to answer various sensitivity queries regarding strong connectivity in $G$ under an arbitrary edge or vertex failure.
For instance, they can compute in $\bigo(n)$ time the strongly connected components (SCCs) that remain in $G$ after the deletion of an edge or a vertex, or report various statistics such as the number of SCCs in constant time per query (failed) edge or vertex.
This result, however, cannot be applied for the solution of the $2$-reachability  problem. The reason is that if the deletion of an edge $e$ leaves two vertices
$u$ and $v$ in different SCCs in $G \setminus e$, the algorithm of \cite{GIP17:SODA} is not able to distinguish if there is still a path or no path
from $u$ to $v$ in $G \setminus e$.

Previously, King and Sagert~\cite{KING2002150} gave an algorithm that can quickly answer sensitivity queries for reachability in
a directed acyclic graph (DAG)~\cite{KING2002150}.
Specifically, they show how to process a DAG $G$ so that, for any pair of query vertices $x$ and $y$, and a query edge $e$, one can test in constant time if there is a path from $x$ to $y$ in $G \setminus e$. Note that the result of  King and Sagert does not yield an efficient solution to the all-pairs $2$-reachability problem, since we need $\bigo(m)$ queries just to find if there is a separating edge for a single pair of vertices. Moreover, their preprocessing time is $\bigo(n^3)$.

Another interesting fact that arises from our work is that, somewhat surprisingly, computing all dominator trees in dense graphs is currently faster than computing a spanning arborescence
from each vertex. The best algorithm for this problem is given by Alon et al.~\cite{alon1992witnesses},
 who studied the problem of constructing a BFS tree from every vertex, and gave an algorithm that runs in $\bigo(n^{(3+\omega)/2})$ time.

Our algorithm uses fast matrix multiplication. Several other important graph-theoretic and network optimization problems can be solved by reductions to fast matrix multiplication.
These include finding maximum weight matchings~\cite{Sankowski:2009},
computing shortest paths~\cite{Zwick:2002},
and finding least common ancestors in DAGs~\cite{Czumaj:2007} and junctions in DAGs~\cite{yuster2008all}.
Our algorithms can be used for constant-time queries on whether there exists a path from vertex $u$ to vertex $v$ avoiding an edge $e$ (called \emph{avoiding path}).
This notion is closely related to a \emph{replacement path}~\cite{Grandoni2012Improved,VassilevskaWilliams:2011,Weimann:2013} (for which we additionally require to be shortest in $G \setminus e$).

\paragraph{Our Techniques.}
Our result is based on two novel approaches, one for DAGs
and one for strongly connected digraphs.
For DAGs we develop an algebra that operates on paths.
We then use some
version of $1$-superimposed coding
to apply our path algebra in a divide and conquer approach.
This allows us
to use Boolean
matrix multiplication, in a similar vein to the computation of transitive closure.
Unfortunately, our algebraic approach does not work for strongly connected digraphs. In this case, we exploit
dominator trees in order to transform a strongly connected digraph $G$ into two auxiliary graphs, so as to
reduce $2$-reachability queries in $G$ to $1$-reachability queries in those auxiliary graphs. This reduction works only for strongly connected digraphs and does not carry over to general digraphs.
Our algorithm for general digraphs is obtained via a careful combination of those two approaches.

\paragraph{Organization.}
The remainder of the paper is organized as follows.
After introducing some basic definitions and notation in Section~\ref{sec:notation}, we present our algorithm in three steps.
In Section~\ref{sec:DAGs} we describe our approach for acyclic graphs, Section~\ref{sec:SCC} covers strongly connected graphs and Section~\ref{sec:generalgraphs} describes their combination for arbitrary digraphs.
We provide a matching lower bound and extend our approach to vertex-disjointness in Sections~\ref{sec:lowerbounds}~and~\ref{sec:vertexdisjoint}, respectively.
Finally, Section~\ref{sec:applications} lists several applications of our algorithm.

\section{Preliminaries}
\label{sec:notation}
We assume that the reader is familiar with standard graph terminology, as contained for instance in~\cite{CLRS}.
Let $G=(V,E)$ be a directed graph (digraph). Given an edge $e=(x,y)$ in $E$, we denote $x$ (resp., $y$) as the \emph{tail}  (resp., \emph{head}) \emph{of $e$}.
A \emph{directed path} in $G$ is a  sequence of vertices $v_1$, $v_2$, $\ldots$, $v_k$, such that edge $(v_i,v_{i+1})\in E$ for $i = 1, 2, \ldots , k-1$. The path is said to contain vertex $v_i$, for $i = 1, 2, \ldots , k$, and edge $(v_i,v_{i+1})$, for $i = 1, 2, \ldots , k-1$.
The \emph{length} of a directed path is given by its number of edges. As a special case, there is a path of length $0$ from each vertex to itself.
We write $\reach{u}{v}$ to denote that there is a path from $u$ to $v$, and $\nreach{u}{v}$ if there is no path from $u$ to $v$.
A \emph{directed cycle} is a directed path, with length greater than $0$, starting and ending at the same vertex.
A \emph{directed acyclic graph} (in short \emph{DAG}) is a digraph with no cycles.
A DAG has a
\emph{topological ordering}, i.e.,  a linear ordering of its vertices such that for every
edge $(u,v)$,
 $u$ comes before $v$ in the ordering (denoted by $u<v$).
A digraph $G$ is \emph{strongly connected} if there is a directed path from each vertex to every other vertex.
The \emph{strongly connected components} of a digraph are its maximal strongly connected subgraphs.
Given a subset of vertices $V'\subset V$, we denote by $G\setminus V'$ the digraph obtained after deleting all the vertices in $V'$, together with their incident edges. Given a subset of
edges $E'\subset E$, we denote by $G\setminus E'$ the digraph obtained after deleting all the edges in $E$'.

\paragraph{2-Reachability and 2-Reachability closure.}
We write $\tedp{u}{v}$ (resp., $\tvdp{u}{v}$) to denote that there are two \emph{edge-disjoint} (resp., internally \emph{vertex-disjoint}) paths from $u$ to $v$, and $\ntedp{u}{v}$ (resp., $\ntvdp{u}{v}$) otherwise. As a special case, we assume that $\tedp{v}{v}$ (resp., $\tvdp{v}{v}$)  for each vertex $v$ in $G$.
We define an abstract set $E^+ = E \cup \{\top,\bot\}$. The semantic of this set is as follows: $e \in E$ corresponds to an edge $e$ separating two vertices, $\top$ corresponds to $\tedp{}{}$ (no single edge  separates) and $\bot$ corresponds to $\nreach{}{}$ (no edge is necessary for separation, vertices are already separated).
Given a digraph $G$, we define the \emph{$2$-reachability closure} of $G$, denoted by $\closure{G}$, to be a matrix such that:
$$
\closure{G}[u,v] \defeq
\begin{cases}
\top \quad& \text{ if } \tedp{u}{v}\\
\bot \quad& \text{ if } \nreach{u}{v}\\
e \quad& \text{ where } e \text{ is any separating edge for } u \text{ and } v.\\
\end{cases}
$$
Since  $\tedp{v}{v}$ for each $v \in V$,
$\closure{G}[v,v]=\top$.
An example of a graph with a $2$-reachability closure matrix is given in Figure~\ref{fig:small_example}.
Note that a 2-reachability closure matrix is not necessarily unique, as there might be multiple separating edges for a given vertex pair.
We define the \emph{$2$-reachability left closure} $\closure{G}_L$ by replacing \emph{any separating edge} with \emph{first separating edge} and the \emph{$2$-reachability right closure} $\closure{G}_R$ by replacing it with \emph{last separating edge}.
\begin{figure}
\begin{minipage}{.49\textwidth}
\centering\includegraphics[scale=1.3]{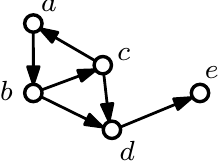}
\end{minipage}
\begin{minipage}{.49\textwidth}
$$\begin{bmatrix} \top & (a,b) & (a,b) & (a,b) & (a,b) \\ (b,c) & \top & (b,c) & \top & (d,e) \\ (c,a) & (c,a) & \top & \top & (d,e) \\ \bot & \bot & \bot & \top & (d,e) \\ \bot & \bot & \bot & \bot & \top \end{bmatrix}$$
\end{minipage}
\caption{A graph and its (non-unique) $2$-reachability closure matrix.}
\label{fig:small_example}
\end{figure}

Note that if there is only one edge separating $u$ and $v$, then $\closure{G}[u,v] = \closure{G}_L[u,v] = \closure{G}_R[u,v]$.
Given any 2-reachability closure matrix, one can compute efficiently the 2-reachability left and right closure matrices. We sketch below the basic idea for the left closure (the right closure being completely symmetric). Let $u$ and $v$ be any two
vertices. If $\closure{G}[u,v]$ is either $\top$ or $\bot$, then $\closure{G}_L[u,v]=\closure{G}[u,v]$. Otherwise, let $\closure{G}[u,v]=(x,y)$:
if $\tedp{u}{x}$ (i.e., if $\closure{G}[u,x]=\top$)
then $(x,y)$ is the first separating edge for $u$ and $v$ and $\closure{G}_L[u,v] = (x,y)$; otherwise, $\ntedp{u}{x}$ (i.e., $\closure{G}[u,x]\neq\top$) and $\closure{G}_L[u,v] = \closure{G}_L[u,x]$.
Algorithm~\ref{alg:leftclosurerecovery}  gives the pseudo-code for computing the $2$-reachability left and right closures $\closure{G}_L$ and $\closure{G}_R$ in a total of
$\bigo(n^2)$ worst-case time.

\begin{algorithm}[t]
\label{alg:leftclosurerecovery}
\caption{Closure recovery}
\KwIn{$N \times N$ matrix $IN$ of $2$-reachability closure, type $\in$ \{\KwLeft, \KwRight\!\!\}}
\KwOut{Matrix $OUT$ of $2$-reachability left closure.}
\SetKwFunction{KwAux}{aux}
\Def{$\KwRecovery(IN, type)$}
{
	$OUT \gets$ $N \times N$  matrix of \KwUnset\;
	\For{$i \gets 1$ \KwTo $N$}
    {
   		\For{$j \gets 1$ \KwTo $N$}
		{
        	$\KwAux(i,j,IN,OUT,type)$\;
        }
    }
	\Return{$OUT$}\;
}
\Def{$\KwAux(u,v,IN,OUT,type)$}
    {
    	\uIf{$OUT[u][v]$ is \KwUnset}
        {
        	\uIf{$IN[u][v] == \top$ \emph{\textbf{or}} $IN[u][v] == \bot$}
          {
          	$OUT[u][v] \gets IN[u][v]$\;
          }
          \uElse
          {
          	$(x,y) \gets IN[u][v]$\;
            \uIf{$type == \KwLeft$}
            {
              \uIf{$IN[u][x] == \top$}
              {
              	$OUT[u][v] \gets (x,y)$\;
              }
              \uElse
              {
            		$OUT[u][v] \gets \KwAux(u,x,IN,OUT,type)$\;
              }
            }
            \uElse
            {
              \uIf{$IN[y][v] == \top$}
              {
                $OUT[u][v] \gets (x,y)$\;
              }
              \uElse
              {
                $OUT[u][v] \gets \KwAux(y,v,IN,OUT,type)$\;
              }
            }
          }
        }
        \Return $OUT[u][v]$\;
    }
\end{algorithm}

\section{All-pairs $2$-reachability in DAGs}
\label{sec:DAGs}
In this section, we present our ${\bigo}(n^{\omega}\log n)$ time algorithm for all-pairs $2$-reachability in DAGs.
The high-level idea is to mimic the way Boolean  matrix multiplication can be used to compute the transitive closure of a graph:
recursively along a topological order, combine the transitive closure of the first and the second half of the vertices in a single matrix multiplication. However, while in transitive closure
for each pair $(i,j)$ we have to store only information on whether there is a path from $i$ to $j$, for all-pairs $2$-reachability this is not enough. First, we describe a path algebra, used by our algorithm to operate on paths between pairs of vertices in a concise manner. We then continue with the description of a matrix product-like operation, which is the backbone of our recursive algorithm. Finally, we show how to implement those operations efficiently using some binary encoding and decoding at every step of the recursion.

Before introducing our new algorithm,
we need some terminology.
Let $G=(V,E)$ be a DAG, and let $E_1, E_2$ be a partition of its edge set $E$, $E=E_1\cup E_2$.
We say that a partition is an \emph{edge split} if there is no triplet of vertices $x,y,z$ in $G$ such that $(x,y) \in E_2$ and $(y,z) \in E_1$ simultaneously.
Informally speaking, under such a split, any path in $G$ from a vertex $u$ to a vertex $v$
consists of a sequence of edges from $E_1$ followed by a sequence of edges from $E_2$ (as a special case, any of those sequences can be empty).
We denote the edge split by $G=(V,E_1,E_2)$ (See Figure~\ref{fig:edge_partition}).
We say that vertex $x$ in $G=(V,E_1,E_2)$  is on the \emph{left} (resp., \emph{right}) \emph{side} of the partition if $x$ is adjacent only to edges in $E_1$ (resp., $E_2$).
We assume without loss of generality that the vertices of $G$ are given in a topological ordering $v_1,v_2,\ldots,v_n$.

\subsection{Algebraic approach}
\label{sec:algebraic}
Consider a family of paths $\mathcal{P} = \{P_1,P_2,\ldots,P_{\ell}\}$, all sharing the same starting and ending vertices $u$ and $v$. We would like to distinguish between the following three possibilities:
(i)~$\mathcal{P}$ is empty;
(ii)~at least one edge $e$ belongs to every path $P_i \in \mathcal{P}$; or
(iii)~there is no edge that belongs to all paths in (nonempty) $\mathcal{P}$.
To do that, we  define the \emph{representation} $\repr(\mathcal{P})$:
$$\repr(\mathcal{P}) \defeq  \bigcap_{i = 1}^{\ell} P_i =
\begin{cases}
\mathbb{U} \quad&\text{ if } \mathcal{P} = \emptyset\\
\emptyset \quad&\text{ if no edge belongs to all $P_i$}\\
\{e\in E : e \in P_i,\ 1\leq i\leq\ell\} \quad&\text{ otherwise. }
\end{cases}
$$
where $\mathbb{U}$ denotes the top symbol in the Boolean algebra of sets (i.e., the complement of $\emptyset$).

We also define
a \emph{left representation} $\lrepr(\mathcal{P}) \in E^+$, where $E^+ = E \cup \{\top,\bot\}$,
 as follows:
$$\lrepr(\mathcal{P}) \defeq
\begin{cases}
\bot \quad&\text{ if } \mathcal{P} = \emptyset\\
\top \quad&\text{ if no edge belongs to all $P_i$}\\
e \quad&\text{ such that } e \in P_i, 1\leq i\leq\ell, \text{ and tail(}e\text{) is \emph{minimum} in the topological order}
\end{cases}
$$
A \emph{right representation} $\rrepr(\mathcal{P}) \in E^+$ is defined symmetrically to $\lrepr(\mathcal{P})$, by replacing \emph{minimum} with \emph{maximum}.
If $\lrepr(\mathcal{P})\in E$ (resp.,
$\rrepr(\mathcal{P}) \in E$), we say that $\lrepr(\mathcal{P})$ (resp., $\rrepr(\mathcal{P})$) is the \emph{first} (resp., \emph{last}) \emph{common edge} in $\mathcal{P}$.
Note that if $\mathcal{P}$ is the set of \emph{all the paths} from $u$ to $v$, then $
\repr(\mathcal{P})$ contains all the information about $\closure{G}[u,v]$. Additionally,
$\closure{G}_L[u,v] = \lrepr(\mathcal{P})$ and $\closure{G}_R[u,v] = \rrepr(\mathcal{P})$. With a slight abuse of  notation we also say that $\closure{G}[u,v] \in  \repr(\mathcal{P})$.

\begin{observation}
\label{ob:repr}
Let $G=(V,E_1,E_2)$ be an edge split of a DAG, and
let $u$ and $v$ be two arbitrary vertices in $G$. For $1 \le i \le n$, let  $\mathcal{P}_i = \{ P\subseteq E_1 : P \text{ is a path from }u\text{ to }v_i \}$, and $\mathcal{Q}_i = \{ Q\subseteq E_2 : Q \text{ is a path from }v_i\text{ to }v \}$  (See Figure~\ref{fig:edge_partition}) and let $\mathcal{S}$ be the family of all paths from $u$ to $v$. Then:
$\repr(\mathcal{S}) = \bigcap_{i=1}^{n} \Big(\repr(\mathcal{P}_i) \cup \repr(\mathcal{Q}_i)\Big)$
\end{observation}
A straightforward application of Observation~\eqref{ob:repr}
yields immediately
a polynomial time algorithm for computing $\closure{G}$. However, this algorithm is not very efficient, since the size of $\repr(\mathcal{P})$ can be as large as $(n-1)$.
In the following
we will show how to obtain a faster algorithm, by replacing
$\repr(\mathcal{P})$ with a suitable combination of $\lrepr(\mathcal{P})$ and $\rrepr(\mathcal{P})$.

\begin{figure}
\begin{minipage}{1\textwidth}
\centering\includegraphics[scale=0.7]{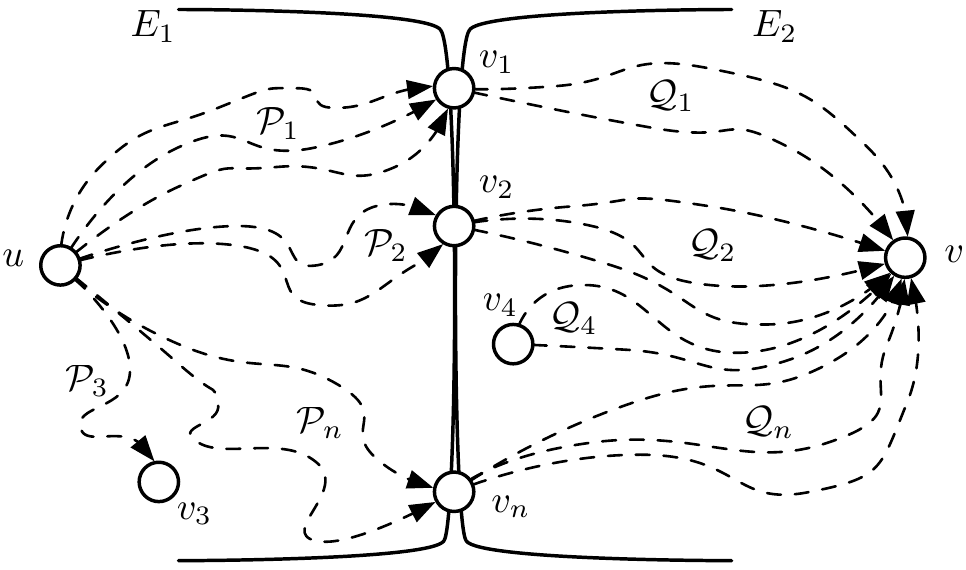}
\end{minipage}
\caption{An edge split of a DAG $G=(V,E_1,E_2)$.}
\label{fig:edge_partition}
\end{figure}
\begin{lemma}
\label{lem:first_edge}
Let $G, \mathcal{P}_i, \mathcal{Q}_i$ and $\mathcal{S}$ be as in Observation~\ref{ob:repr}.
 If $v_i$ is such that both $\reach{u}{v_i}$ in $E_1$ and $\reach{v_i}{v}$ in $E_2$ (both $\mathcal{P}_i \neq \emptyset$ and $\mathcal{Q}_i \neq \emptyset$), then
\begin{enumerate}[(a)]
\item if $\lrepr(\mathcal{S}) = e\in E_1$ and $u\neq v_i$, then $\lrepr(\mathcal{P}_i) = e$;
\item if $\rrepr(\mathcal{S}) = e\in E_2$ and $v_i \neq v$, then $\rrepr(\mathcal{Q}_i) = e$.
\end{enumerate}
\end{lemma}
\begin{proof}
We only prove (a), since (b) is completely analogous.
Assume by contradiction that
$\lrepr(\mathcal{S}) = e\in E_1$, $u\neq v_{i}$ and $\reach{v_i}{v}$ in $E_2$, but
 $\lrepr(\mathcal{P}_{i}) = e' \not= e$.
Since $\lrepr(\mathcal{S}) = e\in E_1$, it must be
$e \in \repr(\mathcal{P}_{i})$, as otherwise we would have a path
$u\leadsto v_i\leadsto v$
avoiding $e$.
Since $e \in \repr(\mathcal{P}_{i})$ and $\lrepr(\mathcal{P}_{i}) = e' \not= e$, all paths in $\mathcal{P}_{i}$ must go first from $u$ to edge $e'$, then to edge $e$ and finally to $v_{i}$. However,
since $\lrepr(\mathcal{S}) = e\in E_1$, then
$e$ is reachable from $u$ by a path avoiding $e'$. By  definition of the edge split, this path must be fully contained in $E_1$, which contradicts
the fact that edge $e'$ precedes $e$ in all paths in $\mathcal{P}_{i}$.
\end{proof}

It is important to note that Lemma~\ref{lem:first_edge} holds regardless of whether $u$ and $v$ are on the same side of the partition or not.

Next, we define two operations, denoted as \emph{serial} and \emph{parallel}.
Although those operations are formally defined on $E^+ = E \cup \{\top,\bot\}$, they have a more intuitive interpretation as operations on path families.
We start with the serial operation $\otimes$. For $a,b \in E^+$, we define:
$$a \otimes b \defeq
\begin{cases}
(\bot,\bot) \quad& \text{ if } a=\bot \text{ or } b=\bot\\
(a,b) \quad& \text{ otherwise.}\\
\end{cases}
$$
We define $\oplus$ as the parallel operator. Namely,
for arbitrary $a \in E^+$:
$a\oplus \bot \defeq a$, $\bot \oplus a \defeq a$,
$a\oplus \top \defeq \top$, $\top \oplus a \defeq \top$,
and otherwise, for $e, e' \in E$:
$$
e \oplus e' \defeq
\begin{cases}
\top \quad& \text{ if } e\not=e'\\
e \quad& \text{ if } e = e'\\
\end{cases}
$$
We extend the definition of $\oplus$ to operate on elements of $E^+ \times E^+$, as follows:
$(a_1,b_1) \oplus (a_2,b_2) \defeq (a_1 \oplus a_2, b_1 \oplus b_2)$.
Ideally, we want the operator $\oplus$ either to preserve consistently the first common edge or to preserve consistently the last common
 edge, under the union of path families. If for instance we preserve the first common edge, that means that if $\mathcal{P}$ and $\mathcal{P}'$ are two path families sharing the same endpoints
 then we want $\lrepr(\mathcal{P} \cup \mathcal{P}') = \lrepr(\mathcal{P}) \oplus \lrepr(\mathcal{P}')$ to hold. However, this is not necessarily the case, as for example both $\mathcal{P}$ and $\mathcal{P}'$ could consist of a single path, with both paths sharing an intermediate edge $e'$, but both with
two different initial edges, respectively $e_1$ and $e_2$. Thus $\lrepr(\mathcal{P}) \oplus \lrepr(\mathcal{P}') = e_1 \oplus e_2 = \top$ while $\lrepr(\mathcal{P} \cup \mathcal{P}') = e'$. As shown in the following lemma, this is not an issue if the path families considered are exhaustive in taking every possible path between a pair of vertices.

\begin{lemma}
\label{lem:path_compression}
Let $G, \mathcal{P}_i, \mathcal{Q}_i$ and $\mathcal{S}$ be as in Observation~\ref{ob:repr}.
Then:
\begin{enumerate}[(a)]
\item $\bigoplus_{i=1}^{n} (\lrepr(\mathcal{P}_i) \otimes
\rrepr(\mathcal{Q}_i)) = (\bot,\bot)$ iff $\repr(\mathcal{S}) = \mathbb{U}$;
\item \label{case_b} if $\bigoplus_{i=1}^{n} (\lrepr(\mathcal{P}_i) \otimes
\rrepr(\mathcal{Q}_i)) = (e_1,\top)$ then $\repr(\mathcal{S}) \ni e_1$;
\item if $\bigoplus_{i=1}^{n} (\lrepr(\mathcal{P}_i) \otimes
\rrepr(\mathcal{Q}_i)) = (\top,e_2)$ then $\repr(\mathcal{S}) \ni e_2$;
\item if $\bigoplus_{i=1}^{n} (\lrepr(\mathcal{P}_i) \otimes
\rrepr(\mathcal{Q}_i)) = (e_1,e_2)$ then $\repr(\mathcal{S}) \ni e_1,e_2$;
\item $\bigoplus_{i=1}^{n} (\lrepr(\mathcal{P}_i) \otimes
\rrepr(\mathcal{Q}_i)) = (\top,\top)$ iff $\repr(\mathcal{S}) = \emptyset$.
\end{enumerate}
\end{lemma}

\begin{proof}
We proceed with a case analysis:
\begin{enumerate}[(a)]
\item
$\repr(\mathcal{S}) = \mathbb{U}$ iff $\mathcal{S} = \emptyset$ iff $\forall_i$ $(\mathcal{P}_i = \emptyset \text{ or } \mathcal{Q}_i = \emptyset)$ iff $\forall_i$ $\lrepr(\mathcal{P}_i) \otimes \rrepr(\mathcal{Q}_i) = (\bot,\bot)$ iff $\bigoplus_{i=1}^{n} (\lrepr(\mathcal{P}_i) \otimes
\rrepr(\mathcal{Q}_i)) = (\bot,\bot)$.
\item
Let $\bigoplus_{i=1}^{n} (\lrepr(\mathcal{P}_i) \otimes
\rrepr(\mathcal{Q}_i)) = (e_1,\top)$.
By definition of $\oplus$, we must have that
$\forall_i (\lrepr(\mathcal{P}_i) \in \{e_1,\bot\} \wedge \rrepr(\mathcal{Q}_i) = \bot)$ and there must be at least one $j$ such that $\lrepr(\mathcal{P}_j) = e_1$ and $\rrepr(\mathcal{Q}_j) \neq \bot$. Hence, any path in $\mathcal{S}$ must contain $e_1$.
\item The proof is similar to (b).
\item The proof is again similar to (b).
\item
If $\bigoplus_{i=1}^{n} (\lrepr(\mathcal{P}_i) \otimes
\rrepr(\mathcal{Q}_i)) \not= (\top,\top)$, then by cases (b), (c) or (d) it follows that $\repr(\mathcal{S}) \not= \emptyset$, clearly a contradiction.

To prove the other direction, assume that $\bigoplus_{i=1}^{n} (\lrepr(\mathcal{P}_i) \otimes
\rrepr(\mathcal{Q}_i)) = (\top,\top)$ and $\repr(\mathcal{S}) \neq \emptyset$.
From case (a) we know that $\repr(\mathcal{S}) \neq \mathbb{B}$, and thus there exists an edge $e\in E$ such that
$e \in \repr(\mathcal{S})$. Without loss of generality, assume that $e \in E_1$.
Then, it must be $\lrepr(\mathcal{S}) = e'$ for some edge $e' \in E_1$. Without loss of generality, assume that $v_1,v_2,\ldots,v_{j}$, are all the vertices such that simultaneously $\reach{u}{v_i}$ in $E_1$, $\reach{v_i}{v}$ in $E_2$ and $u \neq v_i$ (there is at least one such vertex, since $\mathcal{S} \neq \emptyset$ and $\repr(\mathcal{S}) \cap E_1 \neq \emptyset$).
By Lemma~\ref{lem:first_edge}(a), $e'= \lrepr(\mathcal{P}_i)$, $1 \leq i \leq j$. Thus:

$$\bigoplus_{i=1}^{n} (\lrepr(\mathcal{P}_i) \otimes \rrepr(\mathcal{Q}_i)) = \bigoplus_{i=1}^{j} (e',\rrepr(\mathcal{Q}_i)) \oplus\bigoplus_{i=j+1}^{n} (\lrepr(\mathcal{P}_i) \otimes \rrepr(\mathcal{Q}_i))=$$ $$=\bigoplus_{i=1}^{j} (e',\rrepr(\mathcal{Q}_i)) \oplus \bigoplus_{i=j+1}^{n} (\bot,\bot)= (e', \bigoplus_{i=1}^j \rrepr(\mathcal{Q}_i)) \not= (\top,\top),$$
where we have used that (i) if $v_i = u$, then $\mathcal{Q}_i = \emptyset$, as otherwise $\reach{u}{v}$ in $E_2$, with a path avoiding $e \in E_1$, and (ii) by the choice of $j$, for $i > j$, $\lrepr(\mathcal{P}_i) \otimes \rrepr(\mathcal{Q}_i) = (\bot,\bot)$. Thus we have a contradiction.\qedhere
\end{enumerate}
\end{proof}

We now consider the special case where one side of the partition defined in Observation~\ref{ob:repr} contains only paths of length one. In particular, we say that the edge set $E' \subseteq E$ is \emph{thin}, if there exists no triplet of vertices $x,y,z$ such that $(x,y) \in E'$ and $(y,z) \in E'$.

\begin{restatable}{lemma}{singlelayer}
\label{lem:single_layer}
Let $G, \mathcal{P}_i, \mathcal{Q}_i$ and $\mathcal{S}$ be as in Observation~\ref{ob:repr}.
Additionally, let $E_1$ be thin. Then
\begin{enumerate}[(a)]
\item $\bigoplus_{i=1}^{n} (\lrepr(\mathcal{P}_i) \otimes
\rrepr(\mathcal{Q}_i)) = (\bot,\bot)$ iff $\rrepr(\mathcal{S}) = \bot$;
\item if $\bigoplus_{i=1}^{n} (\lrepr(\mathcal{P}_i) \otimes
\rrepr(\mathcal{Q}_i)) = (e_1,\top)$ then $\rrepr(\mathcal{S}) = e_1$;
\item if $\bigoplus_{i=1}^{n} (\lrepr(\mathcal{P}_i) \otimes
\rrepr(\mathcal{Q}_i)) = (\top,e_2)$ then $\rrepr(\mathcal{S}) = e_2$;
\item \label{case_d}  if $\bigoplus_{i=1}^{n} (\lrepr(\mathcal{P}_i) \otimes
\rrepr(\mathcal{Q}_i)) = (e_1,e_2)$ then $\rrepr(\mathcal{S}) = e_2$;
\item $\bigoplus_{i=1}^{n} (\lrepr(\mathcal{P}_i) \otimes
\rrepr(\mathcal{Q}_i)) = (\top,\top)$ iff $\rrepr(\mathcal{S}) = \top$.
\end{enumerate}
\end{restatable}
\begin{proof}
Since $E_1$ is thin, we have that for each $i$: (i) $\lrepr(\mathcal{P}_i) = (u,v_i)$ iff $(u,v_i)\in E_1$, (ii) $\lrepr(\mathcal{P}_i) = \top$ iff $u=v_i$ and (iii)
$\lrepr(\mathcal{P}_i) = \bot$, otherwise. We proceed with a case analysis as in Lemma~\ref{lem:path_compression}.
\begin{enumerate}[(a)]
\item Since $\rrepr(\mathcal{S}) = \bot$ iff $\repr(\mathcal{S}) = \mathbb{U}$, this case follows immediately from Lemma~\ref{lem:path_compression}(a).
\item
The condition implies that there must be $v_j$ such that $e_1 = (u,v_j)$ and $\tedp{v_j}{v}$ in $E_2$. Additionally, for all $i \neq j$ such that $v_i \neq u$, either $(u,v_i) \not\in E_1$ or $\nreach{v_i}{v}$ in $E_2$, and for $v_i=u$ there is $\nreach{v_i}{v}$ in $E_2$. It follows that every path in $G$ from $u$ to $v$ must go through vertex $v_j$, and since $E_1$ is thin, this makes $e_1$ the separating edge. Since $\tedp{v_j}{v}$,  edge $e_1$ is the only possible separating edge for $u$ and $v$. Hence, $\rrepr(\mathcal{S}) = e_1$.

\item
The condition implies that for any $1 \le j \le n$, exactly one of the constraints is satisfied:
(i) $\reach{u}{v_j}$ in $E_1$ (equivalently $v_j = u$ or $(u,v_j) \in E_1$) and
$\rrepr(\mathcal{Q}_j) = e_2$, (ii) $\nreach{u}{v_j}$ in $E_1$ (that is, $v_j \neq u$ and $(u,v_j) \not\in E_1$) or (iii) $\nreach{v_j}{v}$ in $E_2$.
Additionally, unless there exists a $j$ such that $v_j = u$ (which would mean $\tedp{u}{v_j}$), the first constraint is satisfied for at least two distinct values of $j$ since the conditions (ii) and (iii) are not sufficient to satisfy $\bigoplus_{i=1}^{n} (\lrepr(\mathcal{P}_i) \otimes
\rrepr(\mathcal{Q}_i)) = (\top,e_2)$.
It follows that $\rrepr(\mathcal{S}) = e_2$.

\item
The condition implies that for exactly one $j$, there exists an edge $e_1 = (u,v_j)$ and $\reach{v_j}{v}$ in $E_2$, and $\rrepr(\mathcal{Q}_j) = e_2$. Additionally, for every $i \neq j$, either $\nreach{u}{v_j}$ in $E_1$ (that is, $v_j \neq u$ and $(u,v_j) \not\in E_1$) or $\nreach{v_j}{v}$ in $E_2$ (since otherwise there would be a path avoiding $e_1$). Similarly to case (c), it follows that $\rrepr(\mathcal{S}) = e_2$.
\item
Since $\rrepr(\mathcal{S}) = \top$ iff $\repr(\mathcal{S}) = \emptyset$, this case follows immediately from Lemma~\ref{lem:path_compression}(e).\qedhere
\end{enumerate}
\end{proof}

One could prove a symmetric version of Lemma~\ref{lem:single_layer}, with $E_2$ being thin.
However, in the remainder of the paper we stick with Lemma~\ref{lem:single_layer}: namely, we
choose a partition with a thin left side
and thus break case (\ref{case_d}) of Lemma~\ref{lem:single_layer} in favor of the rightmost edge (instead of the leftmost edge, as it would be in the symmetric version).
Consistently, we define the following projection operator $\pi$:
$\pi(\bot,\bot) \defeq \bot$, $\pi(\top,\top) \defeq \top$, $\pi(e',e) = \pi(\top,e) = \pi(e,\top) \defeq e$.
With this new terminology, Lemma~\ref{lem:path_compression} and Lemma~\ref{lem:single_layer}
can be simply restated as follows:
\begin{corollary}
\label{cor:mul1}
Let $G, \mathcal{P}_i, \mathcal{Q}_i$ and $\mathcal{S}$ be as in Observation~\ref{ob:repr}. Then
\begin{enumerate}[(i)]
\item $ \pi(\bigoplus_{i=1}^{n} (\lrepr(\mathcal{P}_i) \otimes
\rrepr(\mathcal{Q}_i)))= \top$ iff $\repr(\mathcal{S}) = \emptyset$,
\item $ \pi(\bigoplus_{i=1}^{n} (\lrepr(\mathcal{P}_i) \otimes
\rrepr(\mathcal{Q}_i)))= \bot$ iff $\repr(\mathcal{S}) = \mathbb{U}$, and
\item
$\pi(\bigoplus_{i=1}^{n} (\lrepr(\mathcal{P}_i) \otimes \rrepr(\mathcal{Q}_i))) \in \repr(\mathcal{S})$ otherwise.
\end{enumerate}
\end{corollary}

\begin{corollary}
\label{cor:mul2}
Let $G, \mathcal{P}_i, \mathcal{Q}_i$ and $\mathcal{S}$ be as in Observation~\ref{ob:repr}, and let $E_1$ be thin. Then
$\pi(\bigoplus_{i=1}^{n} (\lrepr(\mathcal{P}_i) \otimes \rrepr(\mathcal{Q}_i))) = \rrepr(\mathcal{S})$.
\end{corollary}

\paragraph{Matrix product.}
Now we define a \emph{path-based matrix product} based on the previously defined operators:
$(A \circ B)[i,j] \defeq \pi\big( \bigoplus_k A[i,k] \otimes B[k,j] \big)$.
Throughout, we assume  that the vertices of $G$ are sorted according to a topological ordering. In the following lemma, $\mathbf{B}$ represents a thin set of edges (i.e., the set of edges from a subset of vertices to another disjoint subset of vertices).

\begin{lemma}
\label{lem:howtoclosure}
Let $\begin{bmatrix} A & B \\ 0 & C \end{bmatrix}$ be the adjacency matrix of a DAG $G=(V,E)$, where $A,B$ and $C$ are respectively $k \times k$, $k \times (n-k)$ and $(n-k) \times (n-k)$ submatrices. If $\mathbf{B}$ is the matrix containing $\bot$ for every $0$ in $B$ and the appropriate $e \in E$ for every $1$ in $B$, then:
$$ \begin{bmatrix} \closure{A}_L &  \closure{A}_L \circ (\mathbf{B} \circ \closure{C}_R) \\ \bot & \closure{C}_R \end{bmatrix}
$$
is a $2$-reachability closure of $G$ (not necessarily unique).
\end{lemma}

\begin{proof}
Let $V = \{v_1,v_2,\ldots,v_n\}$ be the vertex set in order of rows and columns of the input matrix, and let $V_1 = \{v_1,v_2,\ldots,v_k\}$ and $V_2 = \{v_{k+1}, \ldots, v_n\}$. Matrices $A,B$ and $C$ correspond respectively to all edges from $V_1$ to $V_1$, to all edges from $V_1$ to $V_2$ and to all edges from $V_2$ to $V_2$. We refer to the edge sets represented by those matrices as $E_A, E_B$ and $E_C$.
As a consequence of the fact that there are no edges from $V_2$ to $V_1$, any path from $V_1$ to $V_1$ can use only edges from $E_A$, and any path from $V_2$ to $V_2$ can use only edges from $E_C$.
Thus:
$$\closure{(V,E_A)}_L = \closure{\begin{bmatrix} A & 0 \\ 0 & 0 \end{bmatrix}}_L = \begin{bmatrix} \closure{A}_L & \bot \\ \bot & I \end{bmatrix},$$
$$\closure{(V,E_B)}_L = \closure{\begin{bmatrix} 0 & B \\ 0 & 0 \end{bmatrix}}_L = \begin{bmatrix} I & \mathbf{B} \\ \bot & I \end{bmatrix},$$
$$\closure{(V,E_C)}_R = \closure{\begin{bmatrix} 0 & 0 \\ 0 & C \end{bmatrix}}_R = \begin{bmatrix} I & \bot \\ \bot & \closure{C}_R \end{bmatrix},$$
where
$$I \defeq \begin{bmatrix} \top & \bot & \cdots & \bot \\ \bot & \top & \cdots & \bot \\ \vdots & \vdots & \ddots & \vdots \\ \bot & \bot & \cdots & \top \end{bmatrix}$$

By Corollary~\ref{cor:mul2} (since $E_B$ is thin) and by definition of path-based matrix product:
$$\closure{(V,E_B \cup E_C)}_R = \closure{(V,E_B)}_L \circ  \closure{(V,E_C)}_R = \begin{bmatrix} I & \mathbf{B} \\ \bot & I \end{bmatrix} \circ \begin{bmatrix} I & \bot \\ \bot & \closure{C}_R \end{bmatrix} = \begin{bmatrix} I & \mathbf{B}\circ \closure{C}_R \\ \bot & \closure{C}_R\end{bmatrix}.$$

Finally, by Corollary~\ref{cor:mul1}:
$$\closure{(V,E_A)}_L \circ \closure{(V,E_B \cup E_C)}_R = \begin{bmatrix} \closure{A}_L & \bot \\ \bot & I \end{bmatrix} \circ \begin{bmatrix} I & \mathbf{B}\circ \closure{C}_R \\ \bot & \closure{C}_R\end{bmatrix} = \begin{bmatrix} \closure{A}_L &  \closure{A}_L \circ (\mathbf{B} \circ \closure{C}_R) \\ \bot & \closure{C}_R \end{bmatrix}$$
is a $2$-reachability closure of $G$.
\end{proof}

By Lemma~\ref{lem:howtoclosure},
the 2-reachability closure can be computed by performing path-based matrix products on the left and right 2-reachability closures of smaller matrices. This gives immediately a recursive algorithm for computing the 2-reachability closure: indeed, as already shown in Section~\ref{sec:notation}, one can compute
the left and right 2-reachability closures in $\bigo(n^2)$ time from any 2-reachability closure. In the next section we show how to implement this recursion efficiently by describing how to compute efficiently path-based matrix products.

\subsection{Encoding and decoding for Boolean matrix product}
We start this section by showing how to efficiently compute path-based matrix products using Boolean matrix multiplications.
The first step is to encode each entry of the matrix as a bitword of length $8 k$ where $k = \lceil \log_2 (n+1) \rceil$.
We use Boolean matrix multiplication of matrices of bitwords, with bitwise AND/OR operations, denoted respectively with symbols $\bitwiseAND$ and $\bitwiseOR$. Our bitword length is $\bigo(\log n)$, so matrix multiplication takes $\bigo(n^{\omega} \log n)$ time by performing Boolean matrix multiplication for each coordinate separately.

We make use of the fact that after each multiplication we can afford a post-processing phase, where we perform actions which guarantee that the resulting bitwords represent a valid $2$-reachability closure.

First, we note that when encoding a specific matrix, we know whether it is used as a left-side or as a right-side component of multiplication.
The main idea is to encode left-side and right-side $\bot$ as $\{0\}^{8k}$, left-side and right-side $\top$ as $\{1\}^{8k}$. For any other value, append
$\{1\}^{4k}$ as a prefix or suffix (depending on whether it is used as a left-side or right-side component), to the encoding of an edge. The encoding of an edge is a simple $1$-superimposed code: the concatenation of the edge ID and the complement of the edge ID. To be more precise, whenever a bitword represents an edge $e$ in a left-closure, then it is of the form $\mathrm{ID}_e\overline{\mathrm{ID}_e}\{1\}^{4k}$; whenever a bitword represents an edge $e$ in a right-closure, then it is of the form $\{1\}^{4k}\mathrm{ID}_e\overline{\mathrm{ID}_e}$, where $\overline{w}$ denotes the complement of bitword $w$. The implementation of this encoding is given in Algorithm~\ref{alg:leftsideenc}.

\begin{algorithm}[t]
\label{alg:leftsideenc}
\caption{Left- and right-side encoding}
\KwIn{Matrix $IN$ of dimension $N \times M$ with elements in $E^+$; type $\in$ \{\KwLeft,\KwRight\!\!\}.}
\KwOut{Matrix $OUT$ of dimension $N \times M$ consisting of bitwords.}
\Def{$\KwEncode(IN, type)$}
{
    \ForAll{$IN[i][j]$}
    {
            \uIf{$IN[i][j] == \bot$}
            {
                $OUT[i][j] \gets \{0\}^{8k}$\;
            }
            \uElseIf{$IN[i][j] == \top$}
            {
                $OUT[i][j] \gets \{1\}^{8k}$\;
            }
            \uElse
            {
                \Comment{$IN[i][j] \in E$}
                $(x,y) \gets IN[i][j]$\;
                $b_1 b_2 \ldots b_{k} \gets \text{binary encoding of } x$\;
                $c_1 c_2 \ldots c_{k} \gets \text{binary encoding of } y$\;
                \uIf{ $type == \KwLeft$ }
	       {
	                $OUT[i][j] \gets b_1b_2\ldots b_k \ c_1 c_2 \ldots c_k \overline{b_1 b_2 \ldots b_k} \overline{c_1 c_2 \ldots c_k} \{1\}^{4k}$\;
		}
		\uElse
		{
		                $OUT[i][j] \gets \{1\}^{4k} b_1b_2\ldots b_k \ c_1 c_2 \ldots c_k \overline{b_1 b_2 \ldots b_k} \overline{c_1 c_2 \ldots c_k} $\;
		}
            }

    }
    \Return{OUT}\;
}
\end{algorithm}

The serial operator $\otimes$ is implemented by coordinate-wise AND over two bitwords.
Recall that the operator $\otimes$ always has as its first (left) operand an element from a left-closure matrix and as its second (right) operand an element from a right-closure. It is easy to verify that the result of AND is a concatenation of two bitwords of length $4k$ encoding either $\bot,\top$ or $e \in E$.
We observe that $\otimes$ is calculated properly in all cases: (let $e,e_1,e_2\in E, e_1\not= e_2$)
\begin{enumerate}
\item $e\otimes \top = (e,\top)$\ \ since\ \ $\mathrm{ID}_e \overline{\mathrm{ID}_e}\{1\}^{4k}\ \bitwiseAND\ \{1\}^{8k} = \mathrm{ID}_e \overline{\mathrm{ID}_e}\{1\}^{4k}$
\item $\top \otimes e = (e,\top)$\ \ since\ \ $\{1\}^{8k}\ \bitwiseAND\ \{1\}^{4k}\mathrm{ID}_e \overline{\mathrm{ID}_e} = \{1\}^{4k}\mathrm{ID}_e \overline{\mathrm{ID}_e}$
\item $e_1 \otimes e_2 = (e_1,e_2) $\ \ since\ \ $ \mathrm{ID}_{e_1} \overline{\mathrm{ID}_{e_1}}\{1\}^{4k}\ \bitwiseAND\ \{1\}^{4k}\mathrm{ID}_{e_2} \overline{\mathrm{ID}_{e_2}} = \mathrm{ID}_{e_1} \overline{\mathrm{ID}_{e_1}}\mathrm{ID}_{e_2} \overline{\mathrm{ID}_{e_2}}$
\item $e\otimes \bot = \top \otimes \bot =  \bot \otimes \bot = \bot \otimes e = \bot \otimes \top = (\bot,\bot)$\ \ since\ \ $\{0,1\}^{8k}\ \bitwiseAND\ \{0\}^{8k} = \{0\}^{8k}$
\item $\top \otimes \top  = (\top,\top) $\ \ since\ \ $ \{1\}^{8k}\ \bitwiseAND\ \{1\}^{8k} = \{1\}^{8k}$
\end{enumerate}

The parallel operator  $\oplus$ is implemented as coordinate-wise OR over bitwords of length $8k$. Note that all bitwords can be binary representations of pairs of elements in $E^+$ of the form $(e_1,e_2), (e_1,\top), (\top,e_2), (\bot,\bot), (\top,\top)$, since only those forms appear as a result of an $\otimes$ operation.
Recall that $\oplus$ satisfies $(a_1,b_1) \oplus (a_2,b_2) = (a_1 \oplus a_2, b_1 \oplus b_2)$, thus w.l.o.g. it is enough to verify the correctness of the implementation over the first $4k$ bits of encoding.
Observe that all cases, except
when
both bitwords include encoded edges, are managed correctly by the execution of coordinate-wise OR: (let $e\in E$)
\begin{enumerate}
\item $\bot \oplus \bot = \bot $\ \ since\ \ $ \{0\}^{4k}\ \bitwiseOR\ \{0\}^{4k} = \{0\}^{4k}$
\item $\bot \oplus e = e \oplus \bot = e$\ \ since\ \ $\mathrm{ID}_e \overline{\mathrm{ID}_e}\ \bitwiseOR\ \{0\}^{4k} = \mathrm{ID}_e \overline{\mathrm{ID}_e}$
\item $\bot \oplus \top = \top \oplus \bot = \top\ $\ \ since\ \ $ \{1\}^{4k}\ \bitwiseOR\ \{0\}^{4k} = \{1\}^{4k}$
\item $e \oplus \top = \top \oplus e = \top\ $\ \ since\ \ $ \mathrm{ID}_{e} \overline{\mathrm{ID}_{e}}\ \bitwiseOR\ \{1\}^{4k} = \{1\}^{4k}$
\end{enumerate}

\begin{algorithm}[t]
\caption{Decoding}
\label{alg:decoding}
\KwIn{Matrix $IN$ of dimension $N \times M$ consisting of bitwords.}
\KwOut{Matrix $OUT$ of dimension $N \times M$ consisting of elements in $E^+$.}
\Def{$\KwDecode(IN)$}
{
    \ForAll{$IN[i][j]$}
    {

        	$b_1 b_2 \ldots b_{8k} \gets IN[i][j]$\;	
            \uIf{$b_1 b_2 \ldots b_{8k} == \{0\}^{8k}$}
            {
            	$OUT[i][j] \gets \bot$\;
            }
            \uElseIf{$b_{4k+1}b_{4k+2}\ldots b_{6k} == \overline{b_{6k+1} b_{6k+2} \ldots b_{8k}}$}
            {
            	$x \gets \text{binary decoding of } b_{4k+1}b_{4k+2}\ldots b_{5k}$\;
            	$y \gets \text{binary decoding of } b_{5k+1}b_{5k+2}\ldots b_{6k}$\;
                $OUT[i][j] \gets (x,y)$\;
            }
            \uElseIf{$b_{1}b_{2}\ldots b_{2k} == \overline{b_{2k+1} b_{2k+2} \ldots b_{4k}}$}
            {
            	$x \gets \text{binary decoding of } b_{1}b_{2}\ldots b_{k}$\;
            	$y \gets \text{binary decoding of } b_{k+1}b_{k+2}\ldots b_{2k}$\;
                $OUT[i][j] \gets (x,y)$\;
            }
            \uElse
            {
            	$OUT[i][j] \gets \top$\;
            }

    }
	\Return{OUT}\;
}
\end{algorithm}

We are only left to take care of operations of the form $e_1 \oplus e_2$ for $e_1,e_2 \in E$.
According to the definition of the parallel operator $\oplus$, we would like $e_1 \oplus e_2 = e \in E$ iff $e_1 = e_2 = e$ and otherwise $e_1 \oplus e_2 = \top$.
This special case is handled by the fact that we encode edges using $1$-superimposed codes. That is, the binary representation of $\mathrm{ID}_e$ has the property that
$\mathrm{ID}_e[1\ ..\ 2k] = \overline{\mathrm{ID}_e[2k+1\ ..\ 4k]}$.
Moreover, the coordinate-wise OR of two encodings of edges, that is $\mathrm{X} = \mathrm{ID}_{e_1}\  \bitwiseOR\ \mathrm{ID}_{e_2}$, has this property iff $e_1 = e_2$.

Thus in order to successfully decode the result of chained $\oplus$ from coordinate-wise OR, we need to distinguish the following cases (our result is encoded as $\mathrm{X} = \mathrm{X}[1\ ..\ 2k] \mathrm{X}[2k+1\ ..\  4k]$):
\begin{enumerate}
\item $\mathrm{X} = \{0\}^{4k}$, then the result is $\bot$,
\item $\mathrm{X}[1\ ..\ 2k] = \overline{\mathrm{X}[2k+1\ ..\ 4k]}$, then $\mathrm{X}$ is the encoding of the resulting edge,
\item otherwise the result is $\top$.
\end{enumerate}
With all the tools and notation from above, the path-based matrix product over bitwords can be equivalently stated as
$(\closure{A}_L \circ \closure{B}_R)[i,j] \defeq \KwDecode\left( \bitwiseBigOr_k\ \big( \KwEncode(\closure{A}_L[i,k],\KwLeft) \ \bitwiseAND \ \KwEncode(\closure{B}_R[k,j],\KwRight)\big)\right)$,
where the pseudocode for $\KwDecode$ is provided in Algorithm~\ref{alg:decoding}.

To compute the entries of the final path-based matrix product
(before the execution of $\KwDecode$) it suffices to compute the bitwise Boolean matrix product of appropriate bitword matrices.
That is, we apply $\KwEncode$ to $\closure{A}_L$ and $\closure{B}_R$, then we execute the Boolean matrix product for each coordinate separately, concatenate the coordinates of the resulting Boolean matrices into a matrix of bitwords, and finally execute the $\KwDecode$ operation from Algorithm~\ref{alg:decoding}.
This is illustrated in Algorithm~\ref{alg:mult}.

All the tools developed in this section allow us to compute the $2$-reachability closure for DAGs.
Our recursive algorithm follows closely Lemma~\ref{lem:howtoclosure},
 and its implementation in pseudocode is given as Algorithm~\ref{alg:reachability}.
Since we implemented the right-side version of the projection, we have only to be careful to perform first the right multiplication before the left multiplication.

\begin{algorithm}[t]
\caption{Path-based matrix product}
\label{alg:mult}
\KwIn{Matrices $A$ and $B$ of compatible dimension.}
\KwOut{Matrix being a path-based product of inputs.}
\Def{$\KwMul(A,B)$}
{
	\Return{$\KwDecode(\KwEncode(A, \KwLeft) \cdot \KwEncode(B, \KwRight))$}\;
	\Comment{here $\cdot$ denotes coordinate-wise Boolean matrix multiplication}
}
\end{algorithm}

\begin{lemma}
\label{lem:runtime}
Given a DAG with $n$ vertices, Algorithm~\ref{alg:reachability} computes its $2$-reachability closure in time $\bigo(n^{\omega} \log n)$.
\end{lemma}

\begin{proof}
Algorithm~\ref{alg:mult} computes the path-based matrix product of matrices with every dimension bounded by $n$, if the initial graph size was $n_0$, in time $\bigo({n}^{\omega} \log n_0)$, as it needs to compute $\bigo(\log n_0)$ Boolean matrix products, one for each coordinate of the stored bitwords. Closures are computed in time $\bigo(n^2)$.
The recursion that captures the runtime of Algorithm~\ref{alg:reachability} is thus given by the formula
$T(n) =  T(\lfloor n/2 \rfloor) + T(\lceil n/2 \rceil) + \bigo({n}^{\omega} \log n_0)$
which is satisfied by setting $T(n) = \bigo({n}^{\omega} \log n_0)$. The bound follows.
\end{proof}

\begin{algorithm}
\caption{$2$-reachability closure for DAGs}
\label{alg:reachability}
\KwIn{Matrix $G$ of dimension $N \times N$, $G$ is a DAG with topological order $1,2,\ldots,n$.}
\KwOut{$2$-reachability closure of $G$.}
\Def{$\KwClosureDAG(G)$}
{
	\uIf{$N == 1$}
    {
    	\Return{$\begin{bmatrix} \top \end{bmatrix}$}\;
    }
    $N' \gets \lfloor N/2 \rfloor$\;
    $A \gets G[1\ ..\ N'][1\ ..\ N']$\;
    $B \gets G[1\ ..\ N'][(N'+1)\ ..\ N]$, $0$'s replaced with $\bot$ and $1$'s with edge labels\;
    $C \gets G[(N'+1)\ ..\ N][(N'+1)\ ..\ N]$\;
    $\mathit{A'} \gets \KwRecovery(\KwClosureDAG(A),\KwLeft)$\;
    $\mathit{C'} \gets \KwRecovery(\KwClosureDAG(C),\KwRight)$\;
    \Return{$\begin{bmatrix} \mathit{A'} & \KwMul(\mathit{A'},\KwMul(B,\mathit{C'})) \\ 0 & \mathit{C'} \end{bmatrix} $}\;
}
\end{algorithm}

\section{All-pairs $2$-reachability in strongly connected graphs}
\label{sec:SCC}
In this section we focus on strongly connected graphs.
In this case reachability is simple: for any pair of vertices $( u, v ) \in V \times V$ we have $\reach{u}{v}$ in G. But in case that $\ntedp{u}{v}$ in $G$, finding a separating edge that appears in all paths from $u$ to $v$ in $G$ can still be a challenge. We show that we can report such an edge in constant time after $\bigo(n^{\omega})$ preprocessing. The main result of this section is the following theorem.

\begin{theorem}
\label{thm:scc}
The $2$-reachability closure of a
strongly connected graph can be computed  in time $\bigo(n^{\omega})$.
\end{theorem}

Our construction is based on the notion of auxiliary graph and it will be given in
Section \ref{sec:overview}.
A detailed implementation will be provided in Algorithm~\ref{alg:reachability-scc}. Its running time will be analyzed in Lemma~\ref{lem:aux_time} and its correctness hinges on Lemma~\ref{lem:scqueries}.

\subsection{Reduction to two single-source problems}
\label{sec:sccpreliminaries}
Let $G=(V,E)$ be a strongly connected digraph.
Let $s$ be a fixed but arbitrary vertex of $G$.
The proof of the following lemma is immediate.

\begin{lemma}
\label{lem:sOnOneSide}
For any pair of vertices $u$ and $v$: If there is an edge $e \in E(G)$ such that $\nreach{u}{v}$ in ${G \setminus e}$, then either $\nreach{u}{s}$ in ${G \setminus e}$ or $\nreach{s}{v}$ in ${G \setminus e}$.
\end{lemma}

Let $\mathcal{P}_{u,s}$ be the family of all paths from $u$ to $s$ and let $\mathcal{P}_{s,v}$ be the family of all paths from $s$ to $v$.
We denote by $e_u$ the first edge on all paths in $\mathcal{P}_{u,s}$, and by $e_v$ the last edge on all paths in $\mathcal{P}_{s,v}$.
Note that there might be no edge that is on all paths of $\mathcal{P}_{u,s}$: in this case we say that $e_u$ does not exist. If there are several edges on all paths in $\mathcal{P}_{u,s}$, then they are totally ordered, so it is clear which is the \emph{first} edge (similarly for $e_v$ and $\mathcal{P}_{s,v}$).
We now show that in order to search for a separation witness for $( u, v )$, it suffices to focus on $e_u$ and $e_v$.

\begin{lemma}
\label{lem:EuOrEv}
If there is some $e$ such that $\nreach{u}{v}$ in ${G \setminus e}$, then at least one of the following statements is true:
\begin{itemize}
	\item $e_u$ exists and $\nreach{u}{v}$ in ${G \setminus e_u}$.
	\item $e_v$ exists and $\nreach{u}{v}$ in ${G \setminus e_v}$.
\end{itemize}
\end{lemma}
\begin{proof} If $e = e_u$ or $e = e_v$, the claim is trivial.
Otherwise, by Lemma~\ref{lem:sOnOneSide}, we know that $\nreach{u}{s}$ or $\nreach{s}{v}$ in ${G \setminus e}$. Let us assume that $\nreach{u}{s}$ (See Figure~\ref{fig:edge_separator}).
So $e$ lies not only on any path from $u$ to $v$ but also on any path from $u$ to $s$. As $e_u$ is the first common edge of every path from $u$ to $s$, $e_u$ also lies on every path from $u$ to $e$. As all paths from $u$ to $v$ have to go through $e$, they also have to go through $e_u$ and hence $\nreach{u}{v}$ in ${G \setminus e_u}$.
If $\nreach{s}{v}$ in ${G \setminus e}$, we can show that $\nreach{u}{v}$ in ${G \setminus e_v}$ by the same extremality argument for $e_v$.
\end{proof}

\begin{figure}[t]
\begin{minipage}{1\textwidth}
\centering\includegraphics[scale=0.9]{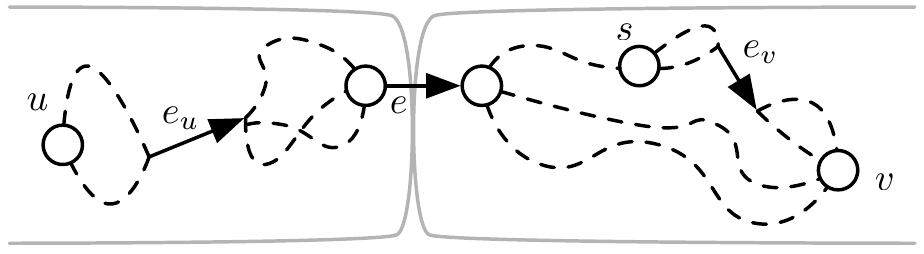}
\end{minipage}
\caption{Illustration of the first case of Lemma~\ref{lem:EuOrEv}.}
\label{fig:edge_separator}
\end{figure}

Hence, in order to check whether there is an edge that separates $u$ from $v$ in $G$, it suffices to look at the reachability information in ${G \setminus e_u}$ (a graph which does not depend on $u$) and at the reachability information in ${G \setminus e_v}$ (a graph which does not depend on $v$). Unfortunately, this is not enough to derive an efficient algorithm, since we would have still to look at as many as $2n$ different graphs (as we explain later, and as it was first shown in \cite{Italiano2012}, there can be at most $2n-2$ edges whose removal can affect the strong connectivity of the graph). As a result, computing the transitive closures of all those graphs would require $\bigo(n^{\omega+1})$ time.
The key insight to reduce the running time to $\bigo(n^{\omega})$ is to construct an auxiliary graph $\auxa$, whose reachability is identical to ${G \setminus e_v}$ for any query pair $( u, v )$, and a second auxiliary graph $\auxb$ whose reachability is identical to ${G \setminus e_u}$ for any query pair $( u, v )$. Note that the edge that is missing from the graph
depends always on one of the two endpoints of the reachability query. As a consequence, we have to consider only $n^2$ and not $n^3$ different queries for $\auxa$ and $\auxb$.

\subsection{Strong bridges and dominator tree decomposition}
\label{sec:dominators}
Before we construct these auxiliary graphs, we need some more terminology and prior results.

\paragraph{Flow graphs, dominators, and bridges.}
A \emph{flow graph} $G_s = (V, E, s)$ is a digraph with a distinguished \emph{start vertex} $s$.
We denote by $G^{R}_s = (V, E^R, s)$ the reverse flow graph of $G_s$; the graph resulted by reversing the direction of all edges $e\in E$.
Vertex $u$ is a \emph{dominator} of a vertex $v$ ($u$ \emph{dominates} $v$) if every path from $s$ to $v$ in $G_s$ contains $u$; $u$ is a \emph{proper dominator} of $v$ if $u$ dominates $v$ and $u \not= v$.
The dominator relation is reflexive and transitive. Its transitive reduction is a rooted tree, the \emph{dominator tree} $D$: $u$ dominates $v$ if and only if $u$ is an ancestor of $v$ in $D$, see Figure~\ref{fig:dominators_example} and Figure~\ref{fig:dominators_rev_example} for examples.
If $v \not= s$, the parent of $v$ in $D$, denoted by $d(v)$, is the \emph{immediate dominator} of $v$: it is the unique proper dominator of $v$ that is dominated by all proper dominators of $v$.
For any vertex $v$, we let $D(v)$ denote the set of descendants of $v$ in $D$, i.e., the vertices dominated by $v$.
Lengauer and Tarjan~\cite{domin:lt} presented an algorithm for computing dominators in  $\bigo(m \alpha(m,n))$ time for a flow graph with $n$ vertices and $m$ edges, where $\alpha$ is a functional inverse of Ackermann's function~\cite{dsu:tarjan}.
Dominators can be computed in linear time
\cite{domin:ahlt,dominators:bgkrtw,dominators:Fraczak2013}.
An edge $(x,y)$ is a \emph{bridge} of the flow graph $G_s$ if all paths from $s$ to $y$ include $(x,y)$.

\begin{figure}[t]
\begin{minipage}{.495\textwidth}
\centering\includegraphics[scale=1.3]{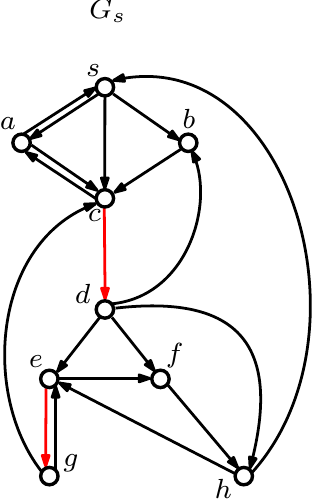}
\end{minipage}
\begin{minipage}{.495\textwidth}
\centering\includegraphics[scale=1.3]{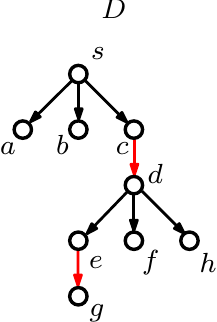}
\end{minipage}
\caption{A flow graph and its dominator tree. Edges marked in red are bridges.}
\label{fig:dominators_example}
\end{figure}

\paragraph{Strong bridges.}
Let $G=(V,E)$ be a strongly connected digraph.
An edge $e$ of $G$ is a \emph{strong bridge} if $G\setminus e$ is no longer
strongly connected.
Let $s$ be an
arbitrary start vertex of $G$.
Since $G$ is strongly connected, all vertices are reachable from $s$ and reach $s$, so we can view both $G$ and $G^R$ as flow graphs with start vertex $s$, denoted respectively by $G_s$ and $G_s^R$.
\begin{property}
\label{property:strong-bridge} \emph{(\cite{Italiano2012})}
Let $s$ be an arbitrary start vertex of $G$. An edge $e=(x,y)$ is a strong bridge of $G$ if and only if it is a bridge of $G_s$
or a bridge of $G_s^R$
(or both).
\end{property}

As a consequence of Property~\ref{property:strong-bridge},
all the strong bridges of the digraph $G$ can be obtained from the bridges of the flow graphs $G_s$ and $G_s^R$, and thus there can be at most $(2n-2)$ strong bridges in a digraph $G$.
Using the linear time algorithms for computing dominators, we can thus compute all strong bridges of $G$ in time $\bigo(m+n) \subseteq \bigo(n^\omega)$.
We  use the following lemma from \cite{2ECB} that holds for a flow graph $G_s$ of a strongly connected digraph $G$.

\begin{lemma}
\label{lemma:partition-paths} \emph{(\cite{2ECB})}
Let $G$ be a strongly connected digraph and let $(x,y)$ be a strong bridge of $G$. Also, let $D$ and $D^R$ be the dominator trees of the corresponding flow graphs $G_s$ and $G_s^R$, respectively, for an arbitrary start vertex $s$.
\begin{enumerate}[(a)]
\item Suppose $x=d(y)$. Let $w$ be any vertex that is not a descendant of $y$ in $D$. Then there is a path from $w$ to $x$ in $G$ that does not contain any proper descendant of $y$ in $D$. Moreover, all simple paths in $G$ from $w$ to any descendant of $y$ in $D$ must contain the edge $(d(y),y)$.
\item Suppose $y=d^R(x)$. Let $w$ be any vertex that is not a descendant of $x$ in $D^R$. Then there is a path from $x$ to $w$ in $G$ that does not contain any proper descendant of $x$ in $D^R$. Moreover, all simple paths in $G$ from any descendant of $x$ in $D^R$ to $w$ must contain the edge $(x,d^R(x))$.
\end{enumerate}
\end{lemma}

\paragraph{Bridge decomposition.}
After deleting from the dominator trees $D$ and $D^R$ respectively the bridges of $G_s$ and $G_s^R$, we obtain the \emph{bridge decomposition}  of $D$ and $D^R$ into forests $\mathcal{D}$ and $\mathcal{D}^R$.
Throughout this section, we denote by $T_v$ (resp., $T_v^R$) the tree in $\mathcal{D}$ (resp., $\mathcal{D}^R$) containing vertex $v$, and by $r_v$ (resp., $r^R_v$) the root of $T_v$ (resp., $T_v^R$).
Given a digraph $G=(V,E)$, and a set of vertices $S \subseteq V$, we denote by $G[S]$ the subgraph induced by $S$.
In particular, $G[D(r)]$ denotes the subgraph induced by the descendants of vertex $r$ in $D$.

\begin{figure}[t]
\begin{minipage}{.495\textwidth}
\centering\includegraphics[scale=1.3]{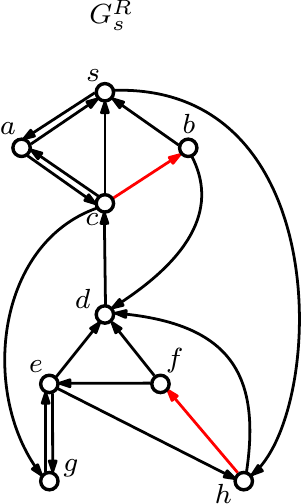}
\end{minipage}
\begin{minipage}{.495\textwidth}
\centering\includegraphics[scale=1.3]{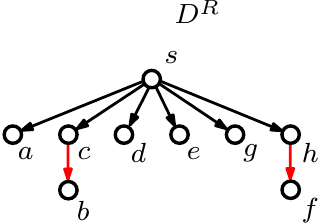}
\end{minipage}
\caption{A reverse flow graph and its dominator tree. Edges marked in red are bridges.}
\label{fig:dominators_rev_example}
\end{figure}

\subsection{Overview of the algorithm and construction of auxiliary graphs}
\label{sec:overview}
The high-level idea of our algorithm is to compute two auxiliary graphs $\auxa$ and $\auxb$ from $G$ and $G^R$, respectively, with the following property:
Given two vertices $u$ and $v$, we have that $\tedp{u}{v}$ in $G$ if and only $\reach{u}{v}$ in $\auxa$ and $\reach{v}{u}$ in $\auxb$.
To construct the auxiliary graphs $\auxa$ and $\auxb$, we use the bridge decompositions of $D$ and $D^R$, respectively.

The two extremal edges $e_u$ and $e_v$,
defined in Section~\ref{sec:sccpreliminaries}, can be also
defined in terms of the bridge decompositions. In particular,
$e_v$ is the bridge entering the tree $T_v$ of the bridge decomposition of $D$, so $e_v = (d(r_v), r_v)$, and $e_u$ is the reverse bridge entering the tree $D^R_u$ of the bridge decomposition of $D^R$, so $e_u = (r_u^R, d^R(r_u^R))$.
Hence if there exists a path from $u$ to $v$ avoiding each of the strong bridges $e_v$ and $e_u$, then
$\tedp{u}{v}$ in $G$.
By Lemma~\ref{lem:EuOrEv},
it is enough if $\auxa$ models the reachability of ${G \setminus e_v}$ and $\auxb$ the reachability of ${G \setminus e_u}$.
So
$\auxa$ is responsible for answering whether $u$ has a path to $v$ avoiding $e_v$, while
$\auxb$ is responsible for answering whether $u$ has a path to $v$ avoiding $e_u$.
Then, if any of the reachability queries in $\auxa$ and $\auxb$ returns false, we immediately have an edge that appears in all paths from $u$ to $v$.

We next show to compute the auxiliary graphs $\auxa$ and $\auxb$ in $\bigo(n^2)$ time. In particular, the auxiliary graph $\auxa = (V,E')$ of the flow graph $G_s = (V,E,s)$ is constructed as follows.
Initially, $E' = E \setminus BR$, where $BR$ is the set of bridges of $G_s$.
For all bridges $(p,q)$ of $G_s$ do the following:
For each edge $(x,y) \in E$ such that $x \in D(q), y \notin D(q)$, we add the edge $(p,y)$ in $E'$, i.e., we set $E' = E' \cup (p,y)$.
A detailed implementation is provided in Algorithm~\ref{alg:reachability-scc}.
Together with graph $\auxa$, the algorithm outputs an array of edges (``witnesses'') $W$,
such that for each vertex $v \not= s$, $W[v]=(d(r_v),r_v)$ is a candidate separating edge for $v$ and any other vertex.
The computation of $\auxb$ is completely analogous.

\SetKwFunction{KwChildren}{children}
\SetKwFunction{KwTransitiveClosure}{transitiveClosure}
\SetKwFunction{KwauxiliaryGraph}{auxiliaryGraph}

\begin{algorithm}
	\caption{$2$-reachability closure in strongly connected graphs}
	\label{alg:reachability-scc}
	\KwIn{Strongly connected graph $G$ on $N$ vertices.}
	\KwOut{$2$-reachability closure of $G$.}
	\Def{$\KwClosureSCC(G)$}
	{
		$s \gets $ arbitrary vertex of $G$\;
		$H,W \gets \KwauxiliaryGraph(G,s)$\;
		$H',W' \gets \KwauxiliaryGraph(G^{R},s)$\;
		$OUT \gets$ $N \times N$ matrix\;
		\ForAll{$OUT[i][j]$}
		{
				\uIf{$(i,j) \not\in H$}
				{
					$OUT[i][j] \gets W[j]$\;
				}
				\uElseIf{$(j,i) \not\in H'$}
				{
					$OUT[i][j] \gets W'[i]^R$\;
				}
				\uElse
				{
					$OUT[i][j] \gets \top$\;
				}
		}
		\Return{OUT}\;
	}
	\Def{$\KwauxiliaryGraph(G,s)$}
	{
		$H \gets G$\;
		$W,R \gets $ arrays of size $N$\;
		$D \gets$ dominator tree of $G_s$\;
		\For{\emph{tree $T$, rooted at $r$, in bottom-up order of bridge decomposition of }$D$}
		{
			$R[r] \gets \bigcup \big\{R[r'] : {d(r') \in V(T)}\big\}$ \quad\quad\Comment{$r'$ is root of children component of $T$}
			$D[r] \gets V(T) \cup \bigcup \big\{D[r'] : {d(r') \in V(T)}\big\}$\;
			\For{$(x,y) \in V(T) \times V(G)$}
			{
				\uIf{$(x,y) \in G$}
				{
					$R[r] \gets R[r] \cup \{y\}$\;
				}
			}
			$R[r] \gets R[r] \setminus D[r]$\;
			\uIf{$s \not\in T$}
			{
				$p \gets d(r)$ \quad\quad\Comment{ $(p,r)$ is a bridge connecting parent of $T$ to $T$}
				\For{$y \in R[r]$}
				{
					$H \gets H \cup (p,y)$\;
				}
				\For{$x \in V(T)$}
				{
					$W[x] \gets (p,r)$
				}
				$H \gets H \setminus (p,r)$\;
			}
		}
		\Return{$\KwTransitiveClosure(H)$,W}\;
	}
\end{algorithm}

Once $\auxa$ and $\auxb$ are computed, their
transitive closure can be computed in
$\bigo(n^\omega)$ time, after which reachability queries can be answered in constant time.
Thus, we can preprocess a strongly connected digraph $G$ in total time $\bigo(n^\omega)$ and answer $2$-reachability queries in constant time, as claimed by Theorem~\ref{thm:scc}.

\begin{figure}
\begin{minipage}{.24\textwidth}
\centering\includegraphics[scale=1]{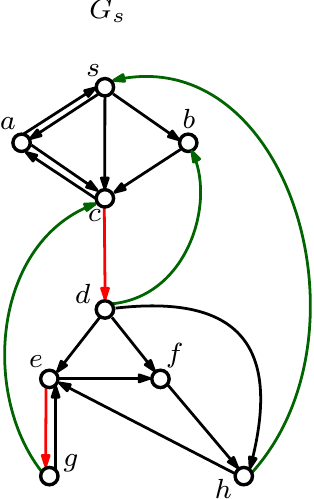}
\end{minipage}
\begin{minipage}{.24\textwidth}
\centering\includegraphics[scale=1]{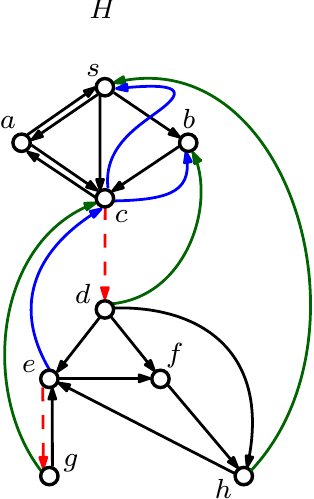}
\end{minipage}
\begin{minipage}{.24\textwidth}
\centering\includegraphics[scale=1]{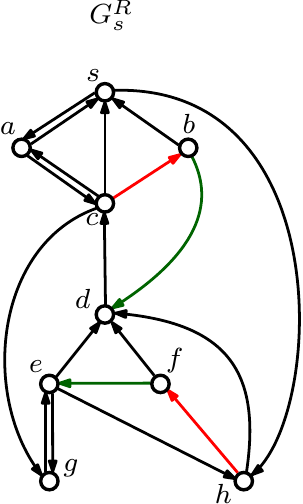}
\end{minipage}
\begin{minipage}{.24\textwidth}
\centering\includegraphics[scale=1]{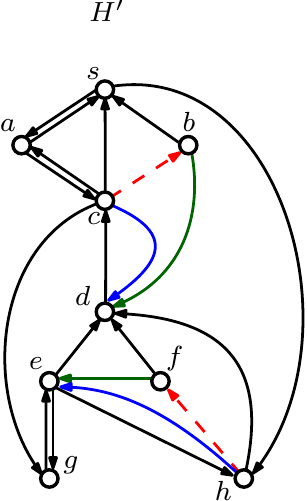}
\end{minipage}
\caption{Auxiliary graphs $\auxa$ and $\auxb$ which are derived from $G_s$ and $G_s^R$, respectively. The deleted edges, the bridges of $G_s$ and $G_s^R$, are shown in red, the newly added edges are shown in blue. The blue edges are drawn along the green edges from $G_s$ and $G_s^R$ which are the reason for their insertion.
Here we see, that for example $b$ is $2$-reachable from $e$, since there are two (edge and vertex) disjoint paths $(e, g, c, d, b)$ and $(e, f, h, s, b)$ in $G$.
In $\auxa$, $e$ reaches $b$ through the path $(e, c, b)$, and in $\auxb$, $b$ reaches $e$ through the path $(b, s, h, e)$. We also see, that edge $(c,d)$ separates $a$ and $f$ in $G$, and even though $f$ reaches $a$ in $\auxb$ through the path $(f, d, c, a)$, $a$ does not reach $f$ in $\auxa$.
To illustrate why both $\auxa$ and $\auxb$ are relevant in Lemma~\ref{lem:scqueries}, consider the following example:
 vertex $c$ is unreachable from $b$ in $G\setminus(b,c)$, which we also detect as there is no $c$-$b$ path in $\auxb$ (even though there is a $b$-$c$ path in $\auxa$).}
\label{fig:aux_example}
\end{figure}

\begin{definition} [Auxiliary graph construction]
\label{def:auxconstruction}
The auxiliary graph $\auxa = (V,E')$ of the flow graph $G_s = (V,E,s)$ is constructed as follows.
Initially, $E' = E \setminus BR$, where $BR$ is the set of bridges of $G_s$.
For all bridges $(p,q)$ of $G_s$ do the following:
For each edge $(x,y) \in E$ such that $x \in D(q), y \notin D(q)$, we add the edge $(p,y)$ in $E'$, i.e., we set $E' = E' \cup (p,y)$.
\end{definition}

\begin{restatable}{lemma}{auxtime}
\label{lem:aux_time}
The auxiliary graph $\auxa$ can be computed in $\bigo(n^2)$ time and space.
\end{restatable}

\begin{proof}
For each root $r$ of a tree $T_r \in \mathcal{D}$, we maintain a set  $R(r) \subseteq V$, initially set to $\emptyset$.
The value $R(r)$
 contains all such endpoints $y$ of edge $(x,y)$ such that $x \in D(r)$ and $y \notin D(r)$.
We process the trees of the bridge decomposition in a bottom-up order of their roots.
For each root $r$ that we visit we compute $R(r)$ in the following tree steps.
First, for each bridge $(p,q)$ of $G_s$ such that $p \in T_r$, we update $R(r)$ by setting $R(r) = R(r)  \cup R(q)$. Second, for every edge $(x,y)$ such that $x\in T_r$ we insert $y$ into $R(r)$.
Finally, we remove all $D(r)$, that is $R(r) = R(r) \setminus D(r)$.
We execute this final step since we are only interested whether there is an edge $(x,y)$ such that $x \in D(r)$ and $y \notin D(r)$.
Clearly, after these steps the set $R(r)$ contains only the desired endpoints.

Note that we actually wish to insert edges to $d(r)$ for each root $r$ of a tree on the bridge decomposition.
Therefore, after computing for each root $r$ its set $R(r)$, we insert to $\auxa$ an edge $(d(r),y)$ for every $y \in R(r)$ (notice that the outgoing edges of $d(r)$ in $G$, except $(d(r),r)$, are also outgoing edges of $d(r)$ in $\auxa$).
Overall, by representing sets as bitmasks, all the $R(r)$ sets can be computed in $\bigo(n^2)$ time.
We spend $\bigo(n^2)$ time in the third step, since we visit each vertex at most one.
Since we traverse every edge only once, the second step takes $\bigo(n+m)$ in total.
The bound follows.
\end{proof}

\begin{restatable}{lemma}{nothingintotheslice}
\label{lem:nothing_into_the_slice}
For all $w \in V$, no edge $(x,y) \in E(\auxa)$ exists with $x \notin D(r_w)$ and $y \in D(r_w)$.
\end{restatable}

\begin{proof}
Assume, by contradiction, that there is an edge $(x,y) \in E(\auxa)$ such that $x \notin D(r_w)$ and $y \in D(r_w)$.
Since $(d(r_w),r_w)$ is a strong bridge in $G_s$, $(x,y)$ does not exist in $G$ (by Lemma~\ref{lemma:partition-paths}).
Hence, by construction, there is an edge $(z,y)\in E(G)$ where $z \in D(x)\setminus x$ and $y \notin D(x)$.
Therefore, $x$ cannot be an ancestor of $w$ in $D$, which implies $z \notin D(r_w)$ and $z\not = d(r_w)$ since $D(z) \cap D(r_w) = \emptyset$.
This is a contradiction, since $(z,y)$, where $z \notin D(r_w)$ and $y\in D(r_w)$, cannot exist in $G$ by Lemma~\ref{lemma:partition-paths}.
\end{proof}

To show the correctness of our approach, we consider queries where we are given an ordered pair of vertices $(u,v)$, and we wish to return whether there exists an edge $e$ such that $\nreach{u}{v}$ in ${G \setminus e}$.
We can answer this query in constant time by answering the queries $\reach{u}{v}$ in $\auxa$ and $\reach{v}{u}$ in $\auxb$.
Given Lemma~\ref{lem:EuOrEv}, it is sufficient to prove the following:

\begin{lemma}
\label{lem:scqueries} The auxiliary graphs $\auxa$ and $\auxb$ satisfy these two conditions:
\begin{itemize}
	\item If $e_v$ exists, then $\reach{u}{v}$ in ${G \setminus e_v}$ if and only if $\reach{u}{v}$ in $\auxa$.
	\item If $e_u$ exists, then $\reach{u}{v}$ in ${G \setminus e_u}$ if and only if $\reach{v}{u}$ in $\auxb$.
\end{itemize}
\end{lemma}

We prove the lemma in two separate steps, one for each direction of the two equivalences.

\begin{lemma}
If $\reach{u}{v}$ in $\auxa$ then $\reach{u}{v}$ in ${G \setminus e_v}$ (and if $\reach{v}{u}$ in $\auxb$ then $\reach{u}{v}$ in ${G \setminus e_u}$).
\end{lemma}

\begin{proof}
By Lemma~\ref{lem:nothing_into_the_slice}, there are no edges incoming into $D(r_v)$ in $\auxa$.
So for a path $P$ from $u$ to $v$ in $\auxa$ to exist, $u$ must lie in $D(r_v)$ and $P$ must lie within $\auxa[D(r_v)]$.
Clearly, if $P$ contains only edges from $E(G[D(r_v)])$, then $P$ is also a valid path from $u$ to $v$ in $G[D(r_v)]$ and thus also in ${G\setminus e_v}$ (recall that $e_v = (d(r_v),r_v)$) and we are done.

Otherwise, we iteratively substitute auxiliary edges of $P$, with paths in $G\setminus e_v$, so that in the end, $P$ is fully contained within ${G \setminus e_v}$.
Let $e^* = (x^*, y^*)$ be the first edge of $P$ such that $e^* \notin E(G[D(r_v)])$, i.e. an auxiliary edge of $\auxa$.
By $e^o = (x^o, y^o)$, we denote the original edge for which we inserted $e^*$ into $\auxa$. Then $x^o \in {D(x^*)\setminus x^*}$ and $y^* = y^o$.
Since all paths from $s$ to $x^o$ contain $x^*$, and all simple paths from $x^*$ to $x^o$ avoid vertices from $V\setminus D(r_v)$ (otherwise, if all paths contained such a vertex $w$ then $s$ would have a path to $x^o$ in $G$ avoiding $x^*$ by Lemma~\ref{lemma:partition-paths}), it follows that $x^*$ has a path $P_{x^*x^o}$ to $x^o$ in $G[D(r_v)]$.
If we now replace $e^*$ in $P$ by $P_{x^*x^o} \cdot e^o$, then $P$ contains a path from $u$ to $y^*$ containing only edges in ${G \setminus e_v}$.
We repeat this argument as long as $P$ contains auxiliary edges and get a path from $u$ to $v$ in ${G \setminus e_v}$.

The statement for $\auxb$ and ${G \setminus e_u}$ can be shown with completely analogous arguments. \qedhere
\end{proof}

\begin{lemma}
If $\reach{u}{v}$ in ${G \setminus e_v}$ then $\reach{u}{v}$ in $\auxa$ (and if $\reach{u}{v}$ in ${G \setminus e_u}$ then $\reach{v}{u}$ in $\auxb$).
\end{lemma}
\begin{proof}
By Lemma~\ref{lemma:partition-paths}~(a), $e_v$ is the only edge in $G$ entering $D(r_v)$.
So any path $P$ from $u$ to $v$ in ${G \setminus e_v}$ can only use edges in $E(G[D(r_v)])$.
If $P$ only contains edges in $H$, we are done.
Otherwise, let $e^* = (x^*, y^*)$ be the first edge of $P$ that is in $E(G[D(r_v)])$ but not in $\auxa$, hence $e^*$ is a bridge of $G_s$.
Let $z^*$ be the first vertex on $P$ after $e^*$ that is not a descendant of $y^*$ in $D$.
Such a vertex $z^*$ exists since $P$ ends at $v$ but $v$ is not a descendant of $y^*$ (recall that $e^*$ is a bridge and lies within $G[D(r_v)])$.
Thus, we can replace the subpath of $P$ between $e^*$ and $z^*$ (including $e^*$) by the edge $(x^*, z^*)$ which is an auxiliary edge of $\auxa$, by the definition of $\auxa$.
We repeat this argument as long as $P$ contains bridges of $G_s$ and get a path from $u$ to $v$ in $\auxa$.

The statement for ${G \setminus e_u}$ and $\auxb$ can be shown with completely analogous arguments. \qedhere
\end{proof}

\section{All-pairs $2$-reachability in general graphs}
\label{sec:generalgraphs}

In this section, we show how to compute
the $2$-reachability of a general digraph by suitably combining the previous algorithms for DAGs and for strongly connected digraphs.
First, note that the $2$-reachability closure of a strongly connected graph $G$ can be constructed as follows: $\closure{G}[i,j] = \top$ if $i$ has two edge-disjoint paths to $j$ and $\closure{G}[i,j] \in E$ if there is an edge $e \in E$ such that $\nreach{i}{j}$ in $G\setminus e$. No entry of $\closure{G}$ contains $\bot$ since $G$ is strongly connected.
After $\bigo(n^\omega)$ time preprocessing all the above queries can be answered in constant time.
Therefore, the $2$-reachability closure can be computed in $\bigo(n^\omega)$ time.

Let $G$ be a general digraph.
The condensation of $G$ is the DAG resulting after the contraction of every strongly connected component of $G$ into a single vertex.
We assume, without loss of generality, that the vertices are ordered as follows: The vertices in the same strongly connected component of $G$ appear consecutively in an arbitrary order, and the strongly connected components are ordered with respect to the topological ordering of the condensation of $G$.
Moreover, we assume that we have access to a function $\KwStrongConnect(u,v)$ that answers whether the vertices $u$ and $v$ are strongly connected.

The key insight is that every idea presented in Section~\ref{sec:DAGs} never truly used the fact that the input graph is a DAG, just the properties of an edge split, that is finding edge partition into two sets so that no vertex has incoming edge from second set and outgoing edge from first set simultaneously. If we are able to extend the definition of an edge split to a general graph in a way highlighted above, and the definitions of $\repr(),\rrepr()$ and $\lrepr()$, then all of the results from Section~\ref{sec:DAGs} carry over to a general graph $G$. Note that given \emph{arbitrary} path family $\mathcal{P}$, $\lrepr(\mathcal{P})$ and $\rrepr(\mathcal{P})$ might be ill-defined, since paths in an arbitrary path family might not share the order of common edges. However, we are only using this notation for path families containing exactly all of paths connecting a given pair of vertices in the graph:  for such families, the order of common edges is shared.

The high-level idea behind our approach
is to extend the
$2$-reachability closure algorithm for DAGs, as follows.
At each recursive call, the algorithm attempts to find a balanced separation of the set of vertices, with respect to their fixed precomputed order, into two sets such that there is no pair across the two sets that is strongly connected.
If such a balanced separation can be found, then the instance is (roughly) equally divided into two instances.
Otherwise,
if there is no balanced separation of the set of vertices into two subsets, then one of the following properties holds: {\it (i)} the larger instance is a strongly connected component, or {\it (ii)} the recursive call on the larger instance  separates a large strongly connected component, on which we can compute the $2$-reachability closure in $\bigo(n^\omega)$ time. We provide pseudocode for this in Algorithm~\ref{alg:reachability-general}.

\begin{algorithm}[t]
	\caption{$2$-reachability closure in general graphs}
		\label{alg:reachability-general}
	\KwIn{Matrix $G$ of dimension $N \times N$, with vertices ordered w.r.t. some fixed topological order of strongly connected components}
	\KwOut{$2$-reachability closure of $G$.}
	\Def{$\KwClosure(G)$}
	{
		\uIf{$N == 1$}
		{
			\Return{$\begin{bmatrix} \top \end{bmatrix}$}\;
		}
		\ElseIf{$ \forall i,j \in \{1,\dots,N\}\ \KwStrongConnect(i,j)==\KwTrue$}
		{
			\Return{$\KwClosureSCC(G)$}\;
		}
		
		$I=\{i: 1\leq i < N,\KwStrongConnect({i},{i+1})== \KwFalse\}$\;
		$N' \gets \arg \min_{i\in I}\{|i- N/2 | \}$\;
		
    $A \gets G[1\ ..\ N'][1\ ..\ N']$\;
    $B \gets G[1\ ..\ N'][(N'+1)\ ..\ N]$, $0$ replaced with $\bot$ and $1$ with edge labels\;
    $C \gets G[(N'+1)\ ..\ N][(N'+1)\ ..\ N]$\;
    $\mathit{A'} \gets \KwRecovery(\KwClosure(A),\KwLeft)$\;
    $\mathit{C'} \gets \KwRecovery(\KwClosure(C),\KwRight)$\;
    \Return{$\begin{bmatrix} \mathit{A'} & \KwMul(\mathit{A'},\KwMul(B,\mathit{C'})) \\ 0 & \mathit{C'} \end{bmatrix} $}\;
	}
\end{algorithm}

\begin{theorem}
\label{thm:runtime2}
Algorithm~\ref{alg:reachability-general} computes $2$-reachability closure of graph on $n$ vertices in time $\bigo(n^\omega \log n)$.
\end{theorem}
\begin{proof}
The algorithm correctness follows from the correctness of Algorithm \ref{alg:reachability} and the fact that Algorithm \ref{alg:reachability-general} separates the input matrix $G$ with dimensions $N\times N$ into  submatrices $A = G[1\ ..\ N'][1\ ..\ N']$, $B = G[1\ ..\ N'][(N'+1)\ ..\ N]$ and $C = G[(N'+1)\ ..\ N][(N'+1)\ ..\ N]$, and such that $\forall i \in [1,N'], j \in [N'+1,N] : G[j,i] = 0$ as required by Lemma \ref{lem:howtoclosure}.
Now we show that  Algorithm \ref{alg:reachability-general} has the same asymptotic running time as Algorithm \ref{alg:reachability}.

The recurrence that provides the running time is the following (we denote by $N_0$ the size of the original graph)
$$
T(N) = T(\min\{N',N-N'\}) + T(\max\{N',N-N'\}) + \bigo(N^\omega \log N_0),
$$
where $N'$ is defined as in Algorithm~\ref{alg:reachability-general}, that is, $N'$ is such $i$ that $\KwStrongConnect({i},{i-1})== \KwFalse$ or $\KwStrongConnect({i},{i+1})== \KwFalse$ and that minimizes $|i - \lfloor N/2 \rfloor |$.
Without loss of generality, assume that $N' \ge N/2$.

Denote by $N'' < N/2$ the start of the strongly connected component that ends at position $N'$. Observe that this component size satisfies $N'-N'' \ge 2(N' - N/2)$. We consider two cases, which intuitively distinguish whether the strongly connected component in the middle of the order is small or large:
\begin{enumerate}
\item If $N' \ge 2/3 N$, then since $N/2 - N'' \ge N'-N/2$ (from the fact that $N'$ is closest to $N/2$), we get $N'' \le N - N' = (N-3/2 N') + N'/2 \le N'/2$. This means that $N''$ is a splitting point in a recursive call on range $[0,N']$, and we get bound
$$T(N) = T(N'') + \bigo((N'-N'')^{\omega}) + T(N-N') + \bigo(N^{\omega} \log N_0) \le T(1/3 N) + T(1/3 N) + \bigo(N^{\omega} \log N_0),$$
where we have used that our claimed runtime bound $T()$ is nondecreasing, so we can use bounds $N'' \le 1/3 N$ and $N-N' \le 1/3 N$.
\item If $N' \le 2/3 N$, then by the fact that $N-N' \le 1/2 N$
$$T(N) = T(N') + T(N-N') + \bigo(N^{\omega} \log N_0) \le T(N/2) + T(2/3 N) + \bigo(N^{\omega} \log N_0).$$
\end{enumerate}

It is easy to see that $T(N) = C \cdot N^{\omega} \log N_0$ satisfies both of the recursive bounds (since $\omega \ge 2$), given large enough constant $C$.

Thus plugging $N_0 = N$ for the total running time yields the claimed bound.
\end{proof}

\section{Matching lower bounds}
\label{sec:lowerbounds}
A simple construction shows that any $2$-reachability oracle for strongly connected graphs can also be used as a reachability oracle for any graph.
Let $G$ be a DAG.
We create a graph $\widehat{G} = (\widehat{V}, \widehat{E})$ from $G$ as follows: we add two new vertices $s$ and $t$ together with the edge $(s,t)$, and for each vertex $v\in V(G)$ we add the edges $(v,s)$ and $(t,v)$.
Clearly, $\widehat{G}$ is strongly connected since we added paths from each vertex $u$ to any other vertex $v$, namely the path $\langle u,s,t,v \rangle$.
All the new paths between vertices in $V(G)$ contain the edge $(s,t)$.
Therefore, a vertex $u$ has two edge-disjoint paths to $v$ in $\widehat{G}$, where $u,v \in V(G)$, if and only if $u$ has a path to $v$ in $G$.

Additionally, for general graphs, there cannot be a significantly faster all-pairs $2$-reachability algorithm (than by a logarithmic factor), as our construction can produce all dominator trees, which by definition encode the necessary information to answer reachability queries in constant time. As computing reachability is asymptotically equivalent to matrix multiplication \cite{fischer1971boolean}, there is no hope to solve all-pairs $2$-reachability in $o(n^\omega)$.

\section{Extension to vertex-disjoint paths}
\label{sec:vertexdisjoint}
Our approach can be modified so that it reports the existence of two vertex-disjoint (rather than edge-disjoint) paths for any pair of vertices.
Although we can formulate the algorithms of Sections \ref{sec:DAGs} and \ref{sec:SCC} so that they use separating vertices (rather than separating edges),
here we sketch how to obtain the same result via a standard reduction, which uses vertex-splitting.
The details of the reduction are as follows. From the original digraph $G=(V,E)$, we compute a modified digraph
$\widehat{G}=(\widehat{V}, \widehat{E})$ by replacing each vertex $v \in V$ by two vertices $v^{+}, v^{-} \in \widehat{V}$, together with the edge $(v^{-}, v^{+}) \in \widehat{E}$,
and replacing each edge $(u,v) \in E$ by $(u^{+}, v^{-}) \in \widehat{E}$. (Thus, $v^{+}$ has the edges leaving $v$, and $v^{-}$ has the edges entering $v$.)
Then, for any pair of vertices $u,v \in V$, $G$ contains two vertex-disjoint paths from $u$ to $v$ if and only if $\widehat{G}$ contains two edge-disjoint paths from $u^{+}$ to $v^{-}$.
Suppose that we apply our algorithm to $\widehat{G}$. Let $u,v$ be any vertices in $G$. If $v$ is reachable from $u$ in $G$ but all paths from $u$ to $v$ in $G$ contain a common vertex,
then the algorithm  reports a separating edge $e \in \widehat{E}$ for the vertices $u^{+}$ and $v^{-}$. If $e = (x^{-}, x^{+})$, then $x$ is a separating vertex for all paths from $u$ to $v$ in $G$.
Otherwise, if $e = (x^{+}, y^{-})$, then $(x,y)$ is a separating edge for all paths from $u$ to $v$ in $G$ and so both $x$ and $y$ are separating vertices.

\section{An application: computing all dominator trees}
\label{sec:applications}
Let $s$ be an arbitrary vertex of $G$.
Recall the bridge decomposition $\mathcal{D}$ of $D$ (Section~\ref{sec:dominators}), which is the forest obtained from $D$ after deleting the bridges of flow graph $G$ with start vertex $s$.
As noted earlier, $T_v$ is the tree in $\mathcal{D}$ that contains vertex $v$, and $r_v$ denotes the root of $T_v$.

We define the \emph{edge-dominator tree} $\widetilde{D}$ of $G$ with start vertex $s$, as the tree that results from $D$ after contracting all vertices in each tree $T_v$ into its root $r_v$. For any vertex $v$ and edge $e=(x,y)$,
$e$ is contained in all paths in $G$ from $s$ to $v$ if and only if $(r_x, r_y)$ is in the path from $s$ to $r_v$ in $\widetilde{D}$.
We denote by $\widetilde{d}(y)$ the parent of a vertex in $\widetilde{D}$. (Both $y$ and $\widetilde{d}(y)$ are roots in $\mathcal{D}$.)

\begin{theorem}
\label{thm:edge_dominator_construction_time}
All sources vertex- and edge-dominator trees can be computed from $\closure{G}_R$ in $\bigo(n^2)$ total time.
\end{theorem}
\begin{proof}
To construct the edge-dominator tree $\widetilde{D}$ with start vertex $s$, we look at the entries $\closure{G}_R[s,v]$, for all vertices $v$.
We have the following cases:
(a) If $\closure{G}_R[s,v] = \bot$ then $v$ is not reachable from $s$ and we set $r_v = \KwUnset$.
(b) If $\closure{G}_R[s,v] = \top$ then we set $r_v = s$.
(c) If $\closure{G}_R[s,v] = (x,y)$ then $(x,y)$ is the last common edge in all paths from $s$ to $v$. We mark $y$, set $r_v = y$, and
temporarily assign $\widetilde{d}(y) = x$ (note that $r_x$ may be unknown at this point).
After we have processed the entries $\closure{G}_R[s, v]$ for all $v$, we make another pass over the marked vertices.
Let $y$ be a marked vertex, for which we temporarily assigned $\widetilde{d}(y) = x$. Then we set $\widetilde{d}(y) = r_x$.
This completes the construction of $\widetilde{D}$, which clearly takes $\bigo(n)$ time.
By repeating this procedure for each vertex in $V$ as a start vertex, we can compute all the edge-dominator trees, each rooted at a different vertex, in total $\bigo(n^2)$ time.
The construction of $D$ is very similar, we apply the standard trick of splitting each vertex into an in- and out-vertex.
\end{proof}

Theorem \ref{thm:edge_dominator_construction_time} enables us to obtain the following results.

\paragraph{Reachability queries after the deletion of an edge or vertex.}
We preprocess each edge-dominator tree $\widetilde{D}$ in $\bigo(n)$ so as to
answer ancestor-descendant relations in constant time~\cite{domin:tarjan}.
We also compute in $\bigo(n)$ time the number of descendants in $D$ of every root $r$ in $\mathcal{D}$.
This allows us to answer various queries very efficiently:
\begin{itemize}
\item Given a pair of vertices $s$ and $t$ and an edge $e=(x,y)$, we can test if $G \setminus e$ contains a path from $s$ to $t$ in constant time.
This is because $e$ is contained in all paths from $s$ to $t$ in $G$ if and only if the following conditions hold:
$e$ is a bridge of flow graph $G$ with start vertex $s$ (i.e., $r_y = y$ and $\widetilde{d}(y)=r_x$) and $y$ is an ancestor of $r_t$ in $\widetilde{D}$.
\item Similarly, given a vertex $s$ and an edge $e=(x,y)$, we can report in constant time how many vertices become unreachable from $s$ if we delete $e$ from $G$.
If $e$ is a bridge of flow graph $G$ with start vertex $s$, then this number is equal to the number of descendants of $y$ in $D$.
\end{itemize}

By computing all vertex-dominator trees of $G$, we can answer analogous queries for vertex-separators.

\paragraph{Computing junctions.}
A vertex $s$ is a \emph{junction} of vertices $u$ and $v$ in $G$, if $G$ contains a path from $s$ to $u$ and a path from $s$ to $v$ that are
vertex-disjoint (i.e., $s$ is the only vertex in common in these paths).  Yuster~\cite{yuster2008all} gave a $\bigo(n^\omega)$ algorithm to compute a single junction for every pair of vertices in a DAG.
By having all dominator trees of a digraph $G$, we can also answer the following queries.
\begin{itemize}
\item Given vertices $s$, $u$ and $v$, test if $s$ is a junction of $u$ and $v$. This is true if and only if $u$ and $v$ are descendants of distinct children of $s$ in $D$.
Hence, we can perform this test in constant time.
\item Similarly, we can report all junctions of a given a pair of vertices in $\bigo(n)$ time. Note that two vertices may have $n$ junctions (e.g., in a complete graph).
\end{itemize}

\paragraph{Computing critical nodes and critical edges.}
Let $G$ be a directed graph.
Define the \emph{reachability function} $f(G)$ as the number of vertex pairs $\langle u, v \rangle$ such that $u$ reaches $v$ in $G$, i.e., there is a directed path from $u$ to $v$.
Here we consider how to compute the \emph{most critical node} (resp., \emph{most critical edge}) of $G$, which is the vertex $v$ (resp., edge $e$) that minimizes
$f(G \setminus v)$ (resp., $f(G \setminus e)$). This problem was considered by Boldi et al.~\cite{Boldi2013} who provided an empirical study of the effectiveness of various heuristics.
A related problem, where we wish to find the vertex $v$ (resp., edge $e$) that minimizes the pairwise strong connectivity of $G \setminus v$ (resp., $G \setminus e$), i.e., $\sum_i {|C_i| \choose 2}$, where $C_i$ are the strongly connected components of $G \setminus v$ (resp., $G \setminus e$), can be solved in linear time ~\cite{GIP17:SODA,PGI17}.

The naive solution to find the most critical node of $G$ is to calculate the transitive closure matrix of $G \setminus v$, for all vertices $v$, and choose the vertex $v$ that
minimizes the number of nonzero elements. This takes $O(n^{\omega + 1})$ time. Similarly, we can
compute the most critical edge in $O(m n^{\omega})$ time. Here we provide faster algorithms that exploit the computation of all dominator trees.
We let $G_u$ and $D_u$ denote, respectively, the flow graph with start vertex $u$ and its dominator tree. Also, we denote by $D_u(v)$ the subtree of $D_u$ rooted at vertex $v$.

\noindent \emph{Computing the most critical node.} Observe that $f(G \setminus v) = \sum_u ( |D_u| - |D_u(v)|)$, since for each vertex $u$, the vertices that become unreachable from $u$ after deleting $v$ are exactly the descendants of $v$ in $D_u$. Hence, we can process all dominator trees in $O(n^2)$ time and compute $|D_u(v)|$ for all vertices $u, v$.
Then, it is straightforward to compute $f(G \setminus v)$, for a single vertex $v$, in $O(n)$ time. Thus, we obtain an algorithm that computes the most critical node of $G$ in
$O(n^{\omega} \log{n})$ total time.

\noindent \emph{Computing the most critical edge.} Almost the same idea works for computing the most critical edge of $G$. Here, we observe that $v$ becomes unreachable from $u$ in $G \setminus e$ if and only if $e=(x,y)$ is a bridge of $G_u$ and $v \in D_u(y)$.
To exploit this observation, we store for each vertex $y$ a list $L_y$ of pairs $\langle u, x \rangle$ such that $(x,y)$ is a bridge in $G_u$. Note that $\langle u, x \rangle \in L_y$ implies that $x$ is the parent of $y$ in $D_u$. Thus, each list $L_y$ has at most $n-1$ pairs.
Now, for each edge $e=(x,y)$, we have $f(G \setminus e) = \sum_{u : \langle u, x \rangle \in L_y} ( |D_u| - |D_u(y)|)$.
So, it is straightforward to compute $f(G \setminus e)$, for a single edge $e$, in $O(n)$ time. To compute $f(G \setminus e)$ for all edges $e$, observe that is suffices to
process only the pairs in all the $L_y$ lists. Specifically, for each edge $(x,y)$ we maintain a count $\mathit{unreach}(x,y)$, initially set to zero.
When we process a pair $\langle u, x \rangle \in L_y$, we increment $\mathit{unreach}(x,y)$ by $|D_u| - |D_u(y)|$.
Clearly, after processing all pairs, the most critical edge is the one with maximum $\mathit{unreach}$ value. Since there are $O(n^2)$ pairs overall in all lists $L_y$,
we obtain an algorithm that computes the most critical edge of $G$ in
$O(n^{\omega} \log{n})$ total time.

\section{Concluding remarks}
\label{sec:conclusion}

In this paper we have shown that the all-pairs $2$-reachability problem can be solved in $\bigo(n^{\omega}\log n)$ time.
Our algorithm produces a witness (separating edge or separating vertex) for every pair of vertices for which $2$-reachability does not hold.
An important  corollary of this result is that we can compute all dominator trees of a digraph within the same time bound. Our work raises some new, and perhaps intriguing, questions.
First, since our algorithms can be used to answer queries on whether there exists a path from vertex $u$ to vertex $v$ avoiding an edge $e$, can we extend our approach to reporting avoiding paths within the same
$\bigo(n^{\omega}\log n)$ (or any sub-cubic) running time?
 Another interesting question is whether one can compute all-pairs $k$-reachability (or equivalently, the existence of $(k-1)$-cuts) in $\widetilde\bigo(n^{\omega})$ time, or even in sub-cubic time, for
$k\geq 3$. It does not seem easy to extend our techniques in this case.

\paragraph{Acknowledgments.}
We would like to thank Paolo Penna and Peter Widmayer for many useful discussions on the problem.

\bibliography{ltg}

\end{document}